\documentclass[10pt, runningheads]{llncs}

\usepackage{float}
\usepackage{amsfonts}
\usepackage{amsmath}
\usepackage{amssymb}
\usepackage{amsfonts}
\usepackage[shortlabels]{enumitem}
\usepackage{mathtools, nccmath}
\usepackage{booktabs}
\usepackage{calligra}
\usepackage{caption}
\usepackage{multirow}
\usepackage{graphicx}
\usepackage{bm}
\usepackage{xfrac}
\usepackage{caption}

\def\nuenc{\nu_\mathsf{fresh}}
\def\nums{\nu_\mathsf{ms}}
\def\numul{\nu_\mathsf{mul}}
\def\U{\mathcal{U}}
\def\V{\mathsf{Var}}
\def\Z{\mathbb{Z}}

\def\E{\mathbb{E}}
\def\Vms{V_{\mathsf{ms}}}

\DeclareMathOperator{\Var}{Var}
\DeclareMathOperator{\Cov}{Cov}

\usepackage[table,usenames,dvipsnames]{xcolor}%
\usepackage[most]{tcolorbox}

\usepackage{hyperref}
\hypersetup{colorlinks = true,
    linkcolor= NavyBlue,
    citecolor = PineGreen,
    linkbordercolor = {Goldenrod}}
    
\usepackage{cleveref}
\usepackage[acronym,shortcuts]{glossaries}
\usepackage{listings}
\usepackage{multirow}
\usepackage{subcaption}
\usepackage{upgreek}
\usepackage{url}
\usepackage{placeins}
\usepackage{adjustbox}

\allowdisplaybreaks

\DeclareMathOperator{\rot}{rot}

\def\nuadd{\nu_\mathsf{add}}
\def\nuenc{\nu_\mathsf{clean}}

\def\numul{\nu_\mathsf{mul}}
\def\numulconst{\nu_\mathsf{const}}

\def\X{\chi}

\def\R{\mathcal{R}}
\def\U{\mathcal{U}}

\def\Z{\mathbb{Z}}

\def\fc{\mathfrak{c}}   
\def\ks{\mathsf{ks}}
\def\pk{\mathsf{pk}}
\def\sk{\mathsf{sk}}

\def\vc{\mathbf{c}}
\def\vd{\mathbf{d}}

\def\fd{\mathfrak{d}}

\def\ncan#1{||#1||^{can}}

\def\Cov{\mathsf{Cov}}
\def\E{\mathbb{E}}
\def\V{\mathsf{Var}}

\def\ms{\mathsf{ms}}

\definecolor{rubblue}{RGB}{0, 53, 96}
\definecolor{rubgreen}{RGB}{141, 174, 16}
\definecolor{rubgray}{RGB}{231, 231, 231}
\definecolor{rubyellow}{RGB}{255, 204, 0}
\definecolor{mygreen}{rgb}{0.400,0.670,0.620} 
\definecolor{mygray}{gray}{0.7}
\definecolor{TIIyellow}{RGB}{255,174,66}
\def\bcB{\color{blue}}

\def\bcR{\color{red}}

\def\ec{\color{black}}
\definecolor{crimson}{rgb}{0.86, 0.08, 0.24}

\def\check#1{{\color{blue} CHECK: #1}}

\allowdisplaybreaks

\usetikzlibrary{calc}
\usetikzlibrary{decorations.pathreplacing}
\usetikzlibrary{math}

\tikzset{add/.style={draw, rubblue, thick, node contents={$+$}, minimum width=5, minimum height=5}}
\tikzset{ciphertext/.style={node contents={$\mathbf{c}_{#1}$}}, ciphertext/.default={}}
\tikzset{const/.style={draw, rubblue, thick, node contents={$\scriptstyle\alpha_{#1}$}, minimum width=5, minimum height=5}}
\tikzset{mul/.style={decorate, decoration={brace, amplitude=2mm}}}
\tikzset{rot/.style={draw, rubblue, thick, node contents={$\scriptstyle\rot$}, minimum width=5, minimum height=5}}
\tikzset{tau/.style={decorate, decoration={brace, amplitude=2mm}}}

\newtheorem{assumption}{Heuristic}
\newtheorem{obs}{Remark}

\lstdefinestyle{rub}{
    basicstyle=\ttfamily\footnotesize,
    captionpos=b,
    commentstyle=\color{rubblue},
    frame=lines,
    keywordstyle=\color{rubblue}\bfseries,
    rulecolor=\color{rubblue}}
\newtcolorbox{algobox}[2][]{
    enhanced,
    sharp corners,
    boxrule=1pt,
    colback=white,
    colframe=rubblue,
    coltitle=rubblue,
    attach boxed title to top left={
        xshift=2mm,
        yshift=-2.5mm},
    boxed title style={
        tile,
        size=minimal,
        colback=white,
        left=1mm, right=1mm,
        before upper=\strut},
    title=#2,
    #1}
    
\author{%
    Beatrice Biasioli\inst{1} \and%
    Chiara~Marcolla\inst{2} \and%
    Nadir~Murru\inst{3} \and%
    Matilda~Urani\inst{4}
    }

\institute{%
    IBM Research Europe - Zurich, Switzerland \& University of Potsdam, Germany \footnote{Part of this work was performed while at Technology Innovation Institute, Abu Dhabi}
    \and%
    Technology Innovation Institute, Abu Dhabi, United Arab Emirates
    \and%
    Università degli studi di Trento, Trento, Italy
    \and%
    Politecnico di Torino, Torino, Italy}
\title{Accurate BGV Parameters Selection: Accounting for Secret and Public Key Dependencies in Average-Case Analysis}

\titlerunning{Accurate BGV Parameters Selection}

\begin{document}
\maketitle
\begin{abstract}
The Brakerski-Gentry-Vaikuntanathan (BGV) scheme is one of the most significant fully homomorphic encryption (FHE) schemes. \\ 
It belongs to a class of FHE schemes whose security is based on the presumed intractability of the Learning with Errors (LWE) problem and its ring variant (RLWE). 
Such schemes deal with a quantity, called \textit{noise}, which increases each time a homomorphic operation is performed. 
Specifically, in order for the scheme to work properly, it is essential that the noise remains below a certain threshold throughout the process.  
For BGV, this threshold strictly depends on the ciphertext modulus, which is one of the initial parameters whose selection heavily affects both the efficiency and security of the scheme. \\
For an optimal parameter choice, it is crucial to accurately estimate the noise growth, particularly that arising from multiplication, which is the most complex operation.
In this work, we propose a novel \textit{average-case} approach that precisely models noise evolution and guides the selection of initial parameters, improving efficiency while ensuring security.
The key innovation of our method lies in accounting for the dependencies among ciphertext errors generated with the same key, and in providing general guidelines for accurate parameter selection that are library-independent.

\end{abstract}

\section{Introduction}
\label{sec:BGVintroduction}
The first Fully Homomorphic Encryption (FHE) scheme was introduced in 2009 by Gentry \cite{gentry2009fully}. 
Since then, several FHE constructions have been proposed, such as BGV \cite{brakerski2014leveled}, BFV \cite{brakerski2012fully,fan2012somewhat}, FHEW \cite{ducas2015fhew}, TFHE \cite{chillotti2016faster,chillotti2020tfhe}, and CKKS \cite{cheon2017homomorphic,cheonRNSHEAAN18}.

The homomorphic encryption schemes currently in use base their security on the presumed intractability of the Learning with Errors (LWE) problem \cite{regev2005lattices}, and its ring variant (RLWE) \cite{lyubashevsky2010ideal}. Informally, the decisional version of RLWE consists of distinguishing polynomial equations 
$(a, b = s\cdot a +  e)\in \R_q\times \R_q$, 
perturbed by small noise $e$ (also called error), from uniform random tuples from $\R_q\times \R_q$, where $\R_q=\Z_q[x]/\langle x^n+1 \rangle$ and $q$ is a positive integer.

Schemes based on the (R)LWE problem face a critical challenge related to the growth of noise during homomorphic operations, which must be carefully controlled to ensure the correct functioning of the encryption scheme. Specifically, the noise must be kept below a certain threshold, which, in the case of BGV, is directly related to the ciphertext modulus parameter $q$. 

As homomorphic operations are performed, the noise increases, and therefore, to maintain the integrity of the scheme, the parameter $q$ must be chosen sufficiently large. However, although increasing $q$ allows for a greater number of operations, it simultaneously compromises both the security and efficiency of the scheme.
Therefore, selecting an appropriate value for $q$ and, in general, determining an optimal set of parameters, is critical. This process requires a balance between security and efficiency while ensuring the correctness of the scheme. \\
One of the key factors in achieving this balance and determining suitable parameters is providing accurate estimates of the error and its growth during the homomorphic operations in the circuit.

This issue is central to research in FHE, and over the years, various approaches have been proposed to address it. As for example, employing the Euclidean norm \cite{brakerski2014leveled}, the infinity norm \cite{fan2012somewhat,kim2021revisiting}, and the canonical norm, also called \textit{worst-case} analysis \cite{costache2016ring,costache2020evaluating,gentry2012homomorphic,iliashenko2019optimisations,mono2022finding}. 
The prevailing trend in the current literature adopts the \textit{average-case analysis}, which involves treating the noise coefficients as random variables distributed according to a Gaussian distribution and studying their expected value and variance.

Interest in this method, initially applied in the TFHE scheme \cite{chillotti2016faster}, and subsequently in the CKKS \cite{costache2022precision,DBLP:journals/cic/BossuatCMNT25}, BGV \cite{costache2023optimisations,murphycentral} and BFV \cite{BFVvar} schemes,  grew due to a recognized discrepancy between the estimates based on worst-case technique and experimental data, as highlighted in \cite{costache2020evaluating}. The introduction of the average-case approach, as seen in \cite{BFVvar,costache2023optimisations}, offers a potential resolution to these disparities, indeed, with this method, it is possible to compute a tight \textit{probabilistic} upper bound.

However, a recent work~\cite{cryptoeprint:2025/1036} criticizes current approaches by showing that the noise coefficients may follow a heavy-tailed distribution, potentially invalidating the assumptions underlying average-case bounds. In particular, focusing on BGV without modulus switching, the authors observe a heavy-tailed behavior of the noise. In the following, we confirm this behavior for the same setting, but we further prove that when modulus switching is applied, this is no longer the case. Specifically, under typical parameter choices used by all major libraries, the noise coefficients in BGV with modulus switching follow a Gaussian distribution, and therefore the average-case noise analysis remains valid. Similarly to our demonstration for BGV, a recent work~\cite{DBLP:journals/cic/BossuatCMNT25} confirms that in CKKS the average-noise analysis remains valid in practice, in line with the observations of~\cite{cryptoeprint:2025/1036}.

It is worth noting that, the heuristics used for the BGV \cite{murphycentral} and CKKS \cite{costache2022precision} schemes often \textit{underestimate} the noise growth due to the assumption of the noises independence, leading to \textit{imprecise bounds}, \ec as also pointed out in \cite{bai2014lattice,costache2022precision,digiusto2025breaking,murphycentral}. 
Such underestimates lead to two potential issues: first, the ciphertext is not correctly decrypted with non-negligible probability due to excessive noise and, second, the scheme is exposed to security vulnerabilities, as shown in recent papers \cite{checri2024practical,cheon2024attacks}.

In light of this, it becomes evident that accounting for the dependencies between the error coefficients is crucial in order to derive increasingly tighter and correct bounds. This, in turn, enables the definition of more accurate operational parameters, making the scheme both more secure and efficient, which is essential for the widespread adoption of FHE.

In this paper, we propose the first average-case noise analysis for BGV that does not provide underestimates, taking into account the dependencies introduced by the common secret and public key. We extend the approach of BFV~\cite{BFVvar}, where the authors introduced a correction function $F$ to adjust the product of variances in homomorphic multiplication, by incorporating the dependency effects of the secret key.
 Similarly, we introduce a correction function which, unlike in~\cite{BFVvar}, is no longer heuristic but derived from formal results presented in our work. Moreover, we observed that in BGV it is necessary to account for additional dependencies introduced by the public key, requiring the correction function to also compensate for these effects. The results obtained in this study suggest that this approach leads to significant improvements in noise analysis.

A related average-case analysis for the BGV scheme is presented in \cite{costache2023optimisations}, where the authors develop a noise estimation method tailored to the specific implementation of BGV in HElib~\cite{halevi2020design}. In contrast, our work proposes a general analysis that does not depend on the specific library and instead focuses on capturing the structural \textit{dependencies} among the errors. We show that considering these dependencies is essential to derive correct, accurate and tighter bounds, independently of specific implementations.

In the BGV scheme, each ciphertext is associated with a \textit{critical quantity} $\nu$ which is a polynomial in $\R$. The critical quantity of a ciphertext $\vc$ defines whether $\vc$ can be correctly decrypted. Specifically, if the size of $\nu$ is below a given bound (depending on $q$) the decryption algorithm works. Otherwise, the plaintext cannot be recovered due to excessive noise growth. Therefore, as previously mentioned, tracking the size of this critical quantity is essential to ensure correct decryption. \\
To provide a clearer picture of what happens to the coefficients of $\nu$, we focus on multiplication, which is the homomorphic operation that highlights most clearly and significantly the dependencies among the critical quantities. 

The BGV public key $\mathsf{pk}\in\R_q\times\R_q$ consists of two polynomials $(-a \cdot s + t e, a)$, where $s$ is the secret key, $t$ is the plaintext modulus, $a\in \R_q$ is randomly chosen and $e\in \R_q$ is the error sampled from a discrete Gaussian distribution $\chi_e$. Roughly speaking, when two ciphertexts are multiplied — even if they were independently computed —  their noises share some common terms which affect the resulting critical quantity $\nu_\mathsf{mult}$. 
More specifically, we observed that the noise in the ciphertexts contains terms that include powers of the secret key $s$ and powers of the term $e$. 
Note that these terms are common to all ciphertexts calculated using the same public key and are responsible for the dependence of the noise.

The paper is structured as follows. Section \ref{sec:BGVpreliminaries} introduces essential definitions and fundamental properties that are necessary for understanding both our contribution and the functioning of the scheme. Section \ref{sec:BGV_scheme} provides a concise overview of the main features and structure of the BGV scheme. In Section \ref{sec:BGVaverage_teo}, we present our key results concerning the behavior of the error term and its growth under homomorphic operations.  Section \ref{sec:BGV_Gaussian} provides a discussion on why, in the BGV scheme, the error coefficients can be modeled as samples from a discrete Gaussian distribution. Section \ref{sec:BGV_circuit} then demonstrates how the findings from Section \ref{sec:BGVaverage_teo} and Section \ref{sec:BGV_Gaussian} can be leveraged to estimate error growth in fixed-operation circuits and to properly select the ciphertext moduli. 
 Section \ref{sec:comparison} compares our approach with state-of-the-art methods, showing how our parameter selection leads to a significant improvement over those currently adopted in major libraries such as OpenFHE and HElib.
Finally, Section \ref{sec:BGVconclusions} concludes the paper and outlines possible directions for future research inspired by our results.

\section{Preliminaries}
\label{sec:BGVpreliminaries}
In this section, we define the general notation and provide the mathematical background that we will use throughout the paper. 
\subsection{Notation}

Let $\Z$ be the ring of integers, and for $d\in\Z_{>0}$ we denote by $\Z_d=\sfrac{\mathbb{Z}}{d\mathbb{Z}}$.
Let $n$ be a power of 2.  We define $\R$ as the ring $\R =\mathbb{Z}[x]/\langle x^n+1\rangle$ and $\R_d = \mathbb{Z}_d[x]/\langle x^n+1\rangle$.

We use $t$ and $q$ to represent the plaintext and the ciphertext modulus, respectively, and $\R_t$ defines the \textit{plaintext space}, where $t$ is chosen such that $t \equiv 1 \mod 2n$. 
Moreover, to define the \textit{ciphertext space}, we need to select $L = M + 1$ moduli, where $M$ is the multiplicative depth of the circuit. Then, for each level $\ell \in \{0,\dots,L-1\}$, we denote 
$$q_\ell = \prod_{j = 0}^{\ell} p_j ,$$
as the \emph{ciphertext modulus} at level $L-1-\ell$. Sometimes, the ciphertext modulo $q_{L-1}$ at level 0 will be indicated simply by $q$.

We use lowercase letters such as $a$ for polynomials and bold letters, like $\bm{a}$, for vectors of polynomials. Moreover, we denote by $a|_i$ the coefficient of $x^i$ of the polynomial $a$. For $a \in \R_d$, we let $[a]_d$ define its \emph{centered reduction modulo $d$}, with coefficients in $[-d/2, d/2)$. Sometimes, we will also write $a \mod d$. 
Unless stated otherwise, we assume that the coefficients of the polynomials in $\R_d$ are always centered modulo $d$.  

Finally, we recall that for \( a,b \in \R \) the \( i \)-th coefficient of their product is given by (see, e.g., \cite{BFVvar})

\begin{equation}\label{productpoly}
    ab|_i = \sum_{j=0}^{n-1} \xi(i,j)a|_jb|_{i-j}, \quad \mbox{ where }\quad  \xi(i,j) =  \left\{ \begin{array}{rl}
        1 \,& \mbox{ for } i-j \in [0, n) \\
     -1 \, & \mbox{ otherwise}.
     \end{array}\right.
\end{equation}

\subsection{Probabilistic distributions}
Given a polynomial $a \in \R$ and a probabilistic distribution $\X$, the notation $a \leftarrow \X$ is used to indicate that each coefficient of $a$ is randomly and independently sampled according to $\X$.
Some distributions that will be frequently considered are the following: 
\begin{itemize}
    \item $\mathcal{U}_q$ as the \textit{uniform distribution} over $\Z_q$ where the representatives modulo $q$ are taken in the interval $\left[-\frac{q}{2},\frac{q}{2}\right)$;
    \item $\mathcal{N}(0,\sigma^2)$ as the normal distribution, also referred to as the \textit{Gaussian distribution}, over $\mathbb{R}$, with mean $0$ and variance $\sigma^2$;
    \item $\mathcal{DG}_{q}(\sigma^2)$ as the \textit{discrete Gaussian distribution}, which involves sampling a value according to $\mathcal{N}(0,\sigma^2)$, rounding it to the nearest integer and then reducing it modulo $q$. Moreover, the representative modulo $q$ is taken in the interval $\left[-\frac{q}{2},\frac{q}{2}\right)$.
\end{itemize}
For the BGV scheme, the notation $\chi_e $ and $\chi_s$ will be adopted in order to indicate the distribution of the error for a RLWE instance and the secret key coefficients, respectively.
Typically, $\chi_e$ is a discrete Gaussian distribution with a suitable standard deviation, while for the secret key, the ternary uniform distribution $\mathcal{U}_3$ is usually considered. However, in some special cases, where bootstrapping is needed, the choice for the secret key distribution falls on the Hamming Weight distribution where the secret key is sampled uniformly among sparse ternary vectors with a fixed number of non-zero entries.

Moreover, we denote by $V_a$ the variance of the coefficients of the polynomial $a$, namely $\V(a|_i)$.
Given $\gamma\in \Z$, if $a,b\in \R$ are two independent polynomials whose coefficients are independent, identically distributed, and have zero mean, then \cite{costache2020evaluating}:
\begin{itemize}
    \item $V_{a+b} = V_a + V_b$
    \item $V_{\gamma a} = \gamma^2V_a$
    \item $V_{a\cdot b} = nV_a\cdot V_b$.
\end{itemize}
Finally, the values of the variance for some common distributions, which will often be employed in the BGV scheme, are
\begin{itemize}
    \item $V_{\mathcal{DG}_q(\sigma^2)} = \sigma^2$ for the discrete Gaussian distribution centered at 0 with standard deviation~$\sigma$;
    \item $V_3 = \frac{2}{3}$ for the ternary distribution $\mathcal{U}_3$;
    \item $V_q = \frac{q^2-1}{12} \approx \frac{q^2}{12}$ for the uniform distribution over integer values in $\left[-\frac{q}{2}, \frac{q}{2}\right)$;
\end{itemize}

\subsection{Infinity and canonical norms}
To conclude this section, we recall the definitions of the infinity and the canonical norm, along with some of their properties.

\begin{definition}
    The \textit{infinity norm} of a polynomial $a \in \R$ is defined as $$\| a \|_{\infty} = \max_{0\le i< n} |\hspace{0.1cm} a|_i \hspace{0.1cm}|.$$
\end{definition}

If $a\in\R_q$ and its coefficients are well-approximated by identically distributed independent Gaussian variables centered at zero, then
\begin{equation}\label{erf}
\mathbb{P} \left( \|a\|_{\infty} > T \right) \leq n \left(1 - \text{erf}\left(\frac{T}{\sqrt{2V_a}}\right) \right),
\end{equation}
where \( \text{erf}(z) \) is the \textit{error function}, defined as $\text{erf}(z) = \frac{2}{\sqrt{\pi}} \int_0^z e^{-t^2} \, dt$, see, e.g., \cite{BFVvar}.
\begin{definition}
    The \textit{canonical embedding norm} of a polynomial $a \in \R$ is
     \begin{equation*}
        \| a \|_{can} = \max_{\substack{1 \leq j < 2n\\ \gcd(j, 2n) = 1}} |a(\zeta^j)| 
        ,
    \end{equation*}
    where $\zeta$ is a primitive $2n$-th root of unity.
    Essentially, it corresponds to the infinity norm of the \textit{canonical embedding} of $a$, denoted as $\sigma(a)$.
\end{definition}
\noindent In our case, for $a,b\in \R$, the relationship between these two norms is given by 
 $|| a  ||_{\infty} \le || a ||_{can}$.
Moreover,  we have \cite{damgaard2012multiparty}:
\begin{align*}
    &|| a\cdot b  ||_{\infty} \hspace{0.12cm} \le n || a  ||_{\infty} \cdot || b  ||_{\infty}  \\
    &|| a\cdot b  ||_{can} \le || a  ||_{can}\cdot|| b  ||_{can} 
\end{align*}

We know that if $a\in\R_q$ is a random polynomial with coefficient variance $V_a$ and $\zeta$ be a primitive $2n^{th}$ root of unity, then the distribution of $a(\zeta)$ is well approximated by a centred Gaussian distribution with variance $nV_a$ \cite{digiusto2025breaking}. 
This immediately translates into a bound on the canonical norm of $a$:
\begin{equation}\label{canonicalnormbound}
||a||^{can} < D \sqrt{n V_a} \,,
\end{equation}
which holds with probability $(1 - e^{-D^2/2})^n \approx 1 - n e^{-D^2/2}$. This means that, for a suitable choice of $D$, the bound fails only with negligible probability \cite{digiusto2025breaking}.
\section{The BGV scheme}
\label{sec:BGV_scheme}
In this section, we recall the three basic encryption functions of the BGV scheme and the homomorphic operations.
\subsection{Basic encryption functions}\label{basicencryption}
\paragraph{Key Generation.}
The key generation function generates  $s \leftarrow \chi_s$ , $a \leftarrow \mathcal{U}_{q_{L-1}}$ and $e \leftarrow \chi_e$ and outputs the \textit{secret key}: $\sk = s\in \R_{q_{L-1}}$,  and the \textit{public key} $\pk = (b,a) \equiv (- a \cdot s + te, a) \mod q_{L-1}$.
\paragraph{Encryption.}
Given the plaintext $m \in \R_t$ and the public key $\pk = (b,a)$, the encryption function outputs the ciphertext $\bm{c} \in \R_{q_{L-1}}^2$ defined as 
    \begin{equation*}
        \bm{c} = (c_0,c_1) \equiv (b \cdot u + te_0 + m, a \cdot u + te_1) \mod q_{L-1},
    \end{equation*}
where $u,e_0,e_1 \in \R_{q_{L-1}}$, with coefficients distributed as $u\leftarrow\chi_s$ and $e_0,e_1\leftarrow\chi_e$. 

\paragraph{Decryption.}
Given a ciphertext $\bm{c}\in \R_{q_{\ell}}^2$ 
and the secret key $\sk = s$, the plaintext is recovered performing the following computations
$m = \left[ \left[ c_0 + c_1 \cdot s \right]_{q_\ell} \right]_t$.

\indent We denote by
    $\nu = [c_0 + c_1 \cdot s]_{q_{\ell}}$ the \textit{critical quantity} corresponding to $\bm{c}$.
In particular, for a fresh ciphertext, it can be rewritten as 
\begin{equation*}
    \nuenc = [c_0 + c_1 \cdot s]_{q_{L-1}} = [ m + t(e \cdot u + e_1 \cdot s + e_0) ]_{q_{L-1}} = [ m + t \epsilon ]_{q_{L-1}},
\end{equation*}
where $\epsilon$ denotes the \textit{error} introduced during encryption.

Considering the reduction modulo $t$ of the critical quantity, it is possible to verify that the plaintext is successfully recovered.
However, if the error is \textit{too large}, the value $m + t\epsilon$ could wrap around the modulus, resulting in an incorrect decryption. So, 
the decryption is correct only if the coefficients of $m + t\epsilon$ remain below a certain threshold.

In addition, the error associated to the ciphertext increases through homomorphic operations \cite{costache2016ring}. Therefore, it is crucial to estimate the magnitude of the critical quantity (called also \textit{noise}) which is typically analyzed using its norm.

\noindent In light of this, to guarantee the correctness of the decryption, the condition on the critical quantity can be expressed as follows \cite{mono2022finding}:
\begin{equation*}
    || \nu ||_{\infty} \le  || \nu ||_{can} < \frac{q_\ell}{2}.
\end{equation*}
Naturally, in order to bound the noise, any type of norm could be used.

Finally, another concept that is often introduced for the estimation of error growth is the \textit{noise budget}, which represents the number of bits remaining before wrap-around would occur.

We will use the term \textit{extended ciphertext} to refer to the tuple $\fc = (\bm{c}, q_\ell , \nu)$ where $\bm{c}$ is the actual ciphertext, $q_\ell$ is the ciphertext modulus, and $\nu$ is the associated critical quantity.
\begin{definition}\label{def.noisebud}
    Let $(\bm{c},q_\ell,\nu)$ be an extended ciphertext. The \textit{noise budget} associated to $\bm{c}$ is the quantity 
    $\log_2(q_\ell) - \log_2(||\nu||) -1$, where $||\cdot ||$ refers to a fixed norm.
\end{definition}

\subsection{Homomorphic Operations}

\paragraph{Ciphertext Addition.}
Let $(\bm{c},q_\ell,\nu)$ and $(\bm{c}',q_\ell,\nu')$ be two extended ciphertexts. 
Their homomorphic sum is
\begin{equation*}
    \mathsf{Add}(\bm{c}, \bm{c'}) := (c_0 + c_0',\; c_1 + c_1') \mod q_\ell.
\end{equation*}
with resulting critical quantity 
$\nu_{\mathsf{add}} = \nu + \nu'$.
\paragraph{Constant Multiplication.}
Let $(\bm{c},q_\ell,\nu)$ be an extended ciphertext and $k \in \R_t$ a fixed polynomial. The homomorphic multiplication by \( k \) is
\begin{equation*}
\mathsf{Mul}_k(\bm{c}) := (k c_0,\; k c_1) \mod q_\ell,
\end{equation*}
and the corresponding critical quantity scales by $k$, yielding $\nu_{\mathsf{const}} = k \cdot \nu$.
\paragraph{Homomorphic multiplication.} 
Let $\bm{c},\bm{c'}$ be two ciphertexts defined in $\R_{q_\ell}$, then
\begin{equation*}
     \mathsf{Mul}(\bm{c}, \bm{c'}) := (d_0, d_1,d_2) = (c_0 \cdot c_0' , c_0 \cdot c_1' + c_1 \cdot c_0', c_1 \cdot c_1' ) \mod q_\ell .
\end{equation*}
As a result, the ciphertext expands from two to three polynomials, which violates the compactness property and makes subsequent operations more costly.  

Recovering the message from $\mathsf{Mul}(\bm{c}, \bm{c'})$ requires computing the reduction modulo $t$ of the resulting critical quantity given by $\nu_{\mathsf{mul}} = d_0 + d_1 s  + d_2 s^2$.

To modify the ciphertext polynomial $d_0 + d_1 s  + d_2 s^2$ back to another polynomial $\bar c_0 + \bar c_1 \cdot s$ encrypting the same plaintext, a technique known as \textit{relinearization}, or \textit{key switching}, is employed.

\paragraph{Key switching.} Intuitively, the key switching technique converts $d_2 s^2$ into $\hat c_0 + \hat c_1 s$ using \textit{somehow} the encryption of $s^2$ under $s$. Indeed, 
$$
\mathsf{Enc}_{s}(s^2)=(\beta,\alpha)  = (- u s + t e + s^2,  u + t e_1  ) \approx (\alpha s + s^2, -\alpha).
$$
Thus, $s^2\approx\beta-\alpha s$ and then $d_0 + d_1 s  + d_2 s^2 \approx d_0 \ec  + d_1 \ec  s  + d_2 (\beta  - \alpha s) = \bar c_0 + \bar c_1 s$.
As expected, while the relinearization step introduces additional noise, it is crucial for ensuring the practicality of the scheme. 

Several key switching techniques have been proposed in the literature, each aiming to optimize this trade-off between correctness and noise growth. The most commonly used are the Brakerski Vaikuntanathan (BV)  variant \cite{brakerski2011fully}, the Gentry Halevi Smart (GHS)  variant \cite{gentry2012homomorphic}, and the Hybrid variant \cite{gentry2012homomorphic}, which can be considered as a combination of the previous ones.
Herein, we do not delve into the details of each method and refer the reader to \cite{mono2022finding} for further information.

\paragraph{Modulus switching.}
The primary aim of the modulus switching technique is to reduce the noise resulting from homomorphic operations. 

Let $(\bm{c}, q_\ell, \nu)$ be the extended ciphertext whose error we aim to reduce, and let $\ell'$ be an integer such that $q_{\ell'} < q_\ell$. The modulus switching procedure outputs $(\bm{c}',q_{\ell'},\nu')$, 
where
\begin{equation*}
    \bm{c}' = \frac{q_{\ell'}}{q_\ell} (\bm{c} + \bm{\delta}) \mod q_{\ell'},
\end{equation*}
with $\bm{\delta} = t[-\bm{c}t^{-1}]_{{q_\ell}/{q_{\ell'}}}$. The $\bm{\delta}$ value can be interpreted as a correction required to ensure that the ciphertext is divisible by ${q_\ell}/q_{\ell'}$ and does not affect the original message. In fact, it only influences the error since $\bm{\delta}\equiv0 \mod t$. Therefore, the new ciphertext $\bm{c}'$ will still decrypt to the original plaintext (scaled by a factor of ${q_\ell}/q_{\ell'}$).

The critical quantity associated to the new ciphertext $\bm{c'}$ can be expressed in terms of that of $\bm{c}$, namely,
\begin{align*}
    \nums &= [c_0'+ c_1' \cdot s]_{q_{\ell'}} =  \frac{q_{\ell'}}{q_\ell}([c_0+c_1 \cdot s]_{q_\ell}+ \delta_0 + \delta_1\cdot s) =  \frac{q_{\ell'}}{q_\ell}(\nu+ \delta_0 + \delta_1\cdot s).
\end{align*}

\section{Average-Case Noise Analysis for BGV}
\label{sec:BGVaverage_teo}
The aim of this section is to investigate the behavior of the noise resulting from the main homomorphic operations supported by the BGV scheme. Throughout this analysis, we consider ciphertexts that are mutually independent, obtained by encrypting independently generated random messages using the same public key.

As previously mentioned, the novel approach introduced in this paper seeks to analyze noise growth by accounting for dependencies among the coefficients of the critical quantities involved. Before delving into the details, it is important to highlight that in BGV, the critical quantity resulting from homomorphic operations can be affected by such dependencies, even when the ciphertexts involved are independent. These dependencies arise because the noise in the ciphertexts includes terms involving powers of the secret key \( s \) and powers of the error term \( e \), which makes it necessary to explicitly consider these contributions when analyzing the variance of the noise.

To study the impact of \( s \) and \( e \), we isolate their contribution in the expression of the critical quantity \( \nu \), using the following notation:
\begin{equation*}
    \nu =  \sum_{\iota} a_{\iota}s^{\iota} = \sum_{\iota}\sum_{\mu} b_{\mu}(\iota) e^{\mu} s^{\iota},
\end{equation*}
where \( a_{\iota} = \sum_{\mu} b_{\mu}(\iota) e^{\mu} \), and \( b_{\mu}(\iota) \) contains no powers of \( s \) or \( e \).

To enhance clarity, for the critical quantity \( \nuenc \) of a fresh ciphertext \( \bm{c} \), this notation yields 
\begin{equation*}
    \nuenc = a_0 + a_1s = b_0(0) + b_1(0)e + b_0(1)s,
\end{equation*}
where
\begin{equation*}
    \begin{cases}
        a_0 = b_0(0) + b_1(0)\cdot e =  (m + te_0) + tu \cdot e \\
        a_1 = b_0(1) = te_1
    \end{cases}
\end{equation*} 
Before introducing our method for studying the growth of error through its variance in fixed circuits, we begin by presenting some considerations and results we have derived regarding the distribution of the coefficients of a generic error term, along with certain properties related to their variances. We then proceed to describe how these properties are used to estimate the variance of the error coefficients after homomorphic multiplications, and finally how such estimates can be applied to circuits consisting of fixed sequences of operations.

\subsection{Mean and Variance Analysis}
 In the following, we prove that the coefficients of the error term are centered at zero. Furthermore, we show that the coefficients of the $b_{\mu}$ terms are uncorrelated, meaning that their pairwise covariance is zero.

\begin{lemma}\label{lemma0}
    Let $\nu = \sum_{\iota} \sum_{\mu} b_{\mu}(\iota) e^{\mu}s^{\iota}$ be the critical quantity associated with a given ciphertext. Then, the following properties hold \begin{itemize}
    \item[a)]  $\Cov(b_{\mu_1}(\iota_1)|_{j_1},b_{\mu_2}(\iota_2)|_{j_2})= 0 $ for $\mu_1 \neq \mu_2 \text{ or } j_1 \neq j_2$, $\forall \iota_1, \iota_2$;
    \item[b)] $\E[b_{\mu}(\iota)|_i] = 0$, $\forall \iota, \mu,i$;
    \end{itemize}
    
\end{lemma}
A proof of Lemma \ref{lemma0} can be found in Appendix \ref{proof.lemma0}
\begin{lemma}\label{lemma1}
 Let $\nu = \sum_{\iota} a_{\iota}s^{\iota}$ represent the critical quantity associated with a given ciphertext, where, for a fixed $\iota$,  $a_{\iota} = \sum_{\mu} b_{\mu}(\iota) e^{\mu}$. Then, the following identity holds
    \begin{align*}
    \V(a_{\iota}s^{\iota}|_i ) = \sum_{\mu \ge 0} \V(b_{\mu}(\iota)|_i) \sum_{k=0}^{n-1} \mathbb{E}[e^{\mu}|_k^2] \sum_{j=0}^{n-1} \mathbb{E}[s^{\iota}|_j^2] .
\end{align*}
Moreover,
   $\V(\nu|_i) = 
   \sum_{\iota,\mu \ge 0} \V(b_{\mu}(\iota)|_i) \sum_{k=0}^{n-1} \mathbb{E}[e^{\mu}|_k^2] \sum_{j=0}^{n-1} \mathbb{E}[s^{\iota}|_j^2].$

\end{lemma}
A proof of \Cref{lemma1} can be found in Appendix \ref{proof.lemma1}.

\subsection{Homomorphic Operations}

We can state our results on the variance computation for operations, especially focusing on the multiplication in the next section.

\begin{proposition}[Encryption]\label{prop.enc}
The invariant noise $\nuenc$ of a fresh ciphertext has coefficient variance 
        \begin{equation}\label{eq.V.enc}
            V_\mathsf{clean} = \V(\nuenc|_i) = \tfrac{t^2}{q^2} \left(\tfrac{1}{12} + nV_eV_u + V_e + nV_eV_s \right).
        \end{equation} 
\end{proposition}
\begin{proof}
The fresh error $\nuenc$ can be written as
$a_0 + a_1s = b_0(0) + b_1(0)e + b_0(1)s$,
where \begin{equation*}
    \begin{cases}
        a_0 = b_0(0) +   b_1(0)\cdot e =  (m + te_0) + tu \cdot e    \\
        a_1 =  b_0(1) = te_1. 
    \end{cases}
\end{equation*}
The proof is thus concluded by applying Lemma~\ref{lemma1} and observing that $\E[s|_j^2] = \V(s|_j) = V_s$ and $\E[e|_j^2] = \V(e|_j) = V_e$ since $\E[s|_j] = \E[e|_j] = 0$.
\end{proof}

\begin{proposition}[Addition \& Constant Multiplication]\label{prop.add}
Let $\alpha \in \R_t$ and $\vc$, $\vc'$ be two independently-computed ciphertexts with invariant noises $\nu,$ and $\nu'$, respectively. Then, the variance of the error coefficients
\begin{itemize}
    \item resulting from the addition of $\vc$ and $\vc'$ is 
        \begin{equation}
            \label{eq.V.add}
            \V(\nuadd|_i) = \V(\nu|_i) + \V(\nu'|_i).
        \end{equation}  
    \item after a multiplication between $\alpha$ and $\vc$ is 
        \begin{equation}
            \label{eq.V.con}
            \V(\numulconst|_i) = \tfrac{(t^2-1)n}{12} \V(\nu|_i).
        \end{equation}
\end{itemize}
\end{proposition}

\begin{proposition}[Modulo Switch] \label{prop.ms}
Let $\fc = (\vc, q_\ell,\nu)$ be an extended ciphertext. The variance of the error coefficients after the modulo switch to the target modulo $q'_\ell$ is 
\begin{equation}\label{eq.V.ms}
    V(\nums|_i)=\tfrac{q_\ell'^2}{q_\ell^2} \V(\nu|_i) + \tfrac{t^2}{12}(1 + nV_s).
\end{equation}
\end{proposition}

Propositions~\ref{prop.add} and~\ref{prop.ms} follow directly from Lemma~\ref{lemma1} by arguments analogous to those in Proposition~\ref{prop.enc}. The proofs are therefore omitted.

We do not explicitly address key switching here, since its negligible impact is discussed in Appendix \ref{KS.negligibility}. 

\subsection{Homomorphic multiplications}

 The main goal of our analysis is to estimate the error growth caused by homomorphic multiplications, the most critical operations to handle. This requires accounting for the dependencies among noise coefficients induced by such multiplications. Our method relies on a \textit{correction function} $F$, whose construction is detailed in this section.  We recall that the function $F$ was introduced in \cite{BFVvar} to ``correct'' the product of the variances, since these are not independent.  Unlike BFV, the BGV scheme demands particular care, as both terms $s$ and $e$ must be considered, whereas the contribution of $e$ is negligible in BFV. Furthermore, unlike \cite{BFVvar}, our construction of $F$ is not heuristic but is based on \Cref{momento_secondo}. 
 We recall the Isserlis theorem which will be exploited in the proof of the Lemma.

\begin{theorem}\cite[Chapter~8]{loeve} \label{teo:iss}
Let $(X_1, \ldots, X_n)$ be a zero-mean multivariate normal random vector, then
\[\E[X_1 \cdots X_n] = \sum \prod_{i,j} \E[X_i X_j] \,,\]
where the sum is over all the partition of $\{1, \ldots, n\}$ into pairs and the $(i, j)$ ranges in these pairs.
\end{theorem}

\begin{lemma}\label{momento_secondo}
    Let $a(x) = a|_0 +a|_1x+ \dots + a|_{n-1}x^{n-1} \in \mathcal{R}$, where $a|_i$'s are i.i.d. random variables with $\E[a|_i] = 0$,  $\E[a|_i^2] = V_a$. \\ 
Then \begin{equation*}
    \E[a^k|_i]  =0 , \quad \E[(a^k|_i)^2]  = k!n^{k-1}V_a^k \,,
\end{equation*}
for all $0 \leq i \leq n-1$, where $a^k|_i$ denotes the $i$--th coefficient of the polynomial $a(x)^k$, for $n$ sufficiently large.
\end{lemma}

\begin{proof}
Without loss of generality, we prove the statement for $a^k|_0$, since the same argument applies to all other terms. 
 Let $\omega_1,\dots,\omega_n$ be the roots of $x^n+1$, which is a cyclotomic polynomial since $n$ is a power of two. Then, we know that  $\sum_{i=1}^{n} \omega_i^{m} = 0$ for $1 \le m \le n-1$, and 
    $a(\omega_1)^k + \dots + a(\omega_n)^k = na^k|_0 + a^k|_1(\omega_1+\dots+\omega_n)+ \dots + a^k|_{n-1}(\omega_1^{n-1}+\dots+\omega_n^{n-1})$,
    thus
\begin{equation*}
    a^k|_0 = \frac{1}{n}(a(\omega_1)^k + \dots + a(\omega_n)^k).
\end{equation*}
%
Now, we define 
    $Z_i := a(\omega_i) = a_0+ a_{1}\omega_i +\dots + a_{n-1}\omega_i^{n-1},$ for $i=1, \ldots, n,$
where, for $n$ sufficiently large, thanks to  \cite[Theorem 9]{digiusto2025breaking} we have that $Z_i$ follows a Gaussian distribution centered at zero. Then 
\begin{align*}
    \E[Z_m \cdot Z_l] &= \E[(a|_0+ a|_{1}\omega_m +\dots + a|_{n-1}\omega_m^{n-1})(a|_0+ a|_{1}\omega_l +\dots + a|_{n-1}\omega_l^{n-1})]= \\ &= \E[a|_0^2] + \E[a|_1^2] \omega_m \omega_l +\dots + \E[a|_{n-1}^2]\omega_m^{n-1}\omega_l^{n-1} = \\ 
    &= V_a(1+\omega_m\omega_l + \dots + \omega_m^{n-1}\omega_l^{n-1})= \\ &= \begin{cases}
        nV_a \quad \quad \text{if} \quad \omega_m = \omega_l^{-1} \\ 
        0 \quad \quad \text{otherwise}.
    \end{cases}
\end{align*}
Therefore, for $m=l$,  $\E[Z_m^2] = 0$ and, since every linear combination of $Z_i$'s is normally distributed, we can apply the Isserlis theorem (\Cref{teo:iss}) to obtain $\E[Z_m^k] = 0$. 
Consequently \begin{equation*}
    \E[a^k|_0] = \frac{1}{n}(\E[Z_1^k]+ \dots + \E[Z_n^k]) =0.
\end{equation*}
Moreover, 
   $ \E[a^k|_0^2] = \frac{1}{n^2} \sum_{m}\E[Z_m^kZ_l^k] = \frac{1}{n^2} \sum_{m,l}\E[Z_m^kZ_{m^{-1}}^k] = k!n^{k-1}V_a^k$,
where the last equality follows again from the Isserlis theorem. Indeed, 
\begin{align*}
    \E[Z_m^kZ_{m^{-1}}^k] &= \E[Z_m\cdot \dots \cdot Z_m \cdot Z_{m^{-1}}\cdot \dots \cdot Z_{m^{-1}}] = \\
     &= \E[X_1\cdot \dots \cdot X_m \cdot X_{m+1}\cdot \dots \cdot X_{2m}] = \\ &= \sum \prod \E[X_i X_j] = \\
     &= \begin{cases}
         0 \quad \quad \text{if} \, \, \, 0 \le i,j \le k \, \, \,\text{or} \, \, \, k+1 \le i,j \le 2k+1 \\
         nV_a \quad \quad \text{otherwise}.
     \end{cases}
\end{align*}
Therefore, $\E[Z_m^kZ_{m^{-1}}^k]$ is given by the sum of $k!$ addends whose value is $n^kV_a^k$, i.e., $\E[Z_m^kZ_{m^{-1}}^k] = k! n^k V_a^k$.
Finally, 
$\E[a^k|_0^2] =\frac{1}{n^2}\cdot n\cdot  k! \cdot  n^k \cdot V_a^k =  k!   n^{k-1}  V_a^k$.
\qed
\end{proof}

\begin{obs} 
This result extends some of \cite[Theorem 4.3]{cryptoeprint:2025/1036}, where the authors proved the same statement, using a different approach, but \textit{only} for the case where the $a|_i$'s are normal independent random variables. In our case, the coefficients of the polynomial $a(x)$ follow \textit{any} distribution (with zero mean and finite variance). Moreover, a proof of the specific case $k=2$ is also provided in~\cite[Lemma 2]{DBLP:journals/cic/BossuatCMNT25}.
\end{obs}

\begin{definition}
    Let $a \in \R$. The correction function $F_a$ is defined as \begin{equation}\label{appF}
    F_a(\iota_1,\iota_2) = \frac{\sum_{i=0}^{n-1} \E[ a^{\iota_1+\iota_2}|_i^2 ] }{\sum_{i_1=0}^{n-1} \E[a^{\iota_1}|_{i_1}^2 ]\sum_{i_2=0}^{n-1} \E[a^{\iota_2}|_{i_2}^2] } \,.
\end{equation}
\end{definition}
From \Cref{momento_secondo}, it follows that for polynomials with coefficients following distributions of that form, it holds that 
\begin{equation}\label{appF2}
    F_a(\iota_1,\iota_2) = \frac{(\iota_1 + \iota_2)! }{\iota_1 !\,\iota_2 !} \,.
\end{equation}
We observe that the value of $F_a$ is the same for all polynomials as in \Cref{momento_secondo}. However, for the sake of clarity in presenting our results, we will keep separate notations, writing 
$F_s$ for the secret key $s$ and $F_e$ for the public key $e$.

\noindent We can now introduce our main theorem. 
\begin{theorem}\label{teorema1}
Let $\nu = \sum_{\iota} a_{\iota}s^{\iota},\nu'= \sum_{\iota} a'_{\iota}s^{\iota}$ be the critical quantities of two independently computed ciphertexts defined with respect to the same modulus $q$. Then 
    \begin{equation*}
    \V((a_{\iota_1}s^{\iota_1}a_{\iota_2}'s^{\iota_2})|_i) \le n \V((a_{\iota_1}s^{\iota_1})|_i)\V((a_{\iota_2}'s^{\iota_2})|_i) F_s(\iota_1,\iota_2) F_e(K_1,K_2),
\end{equation*}
where $K_1,K_2$ represent the highest power of $e$ appearing in  $a_{\iota_1}, a'_{\iota_2}$, respectively.
\end{theorem}
\begin{proof}
    From \Cref{lemma1} it is possible to express the variance of two generic terms $a_{\iota_1} s^{\iota_1}|_i$ and $a_{\iota_2}' s^{\iota_2}|_i$ as  \[
\begin{cases}
    \V(a_{\iota_1} s^{\iota_1}|_i) =  \sum_{\mu_1 = 0}^{K_1} \V(b_{\mu_1}(\iota_1)|_i) \sum_{j_1=0}^{n-1} \E[e^{\mu_1}|^2_{j_1}] \sum_{j_2=0}^{n-1} \E[s^{\iota_1}|^2_{j_2}] \\ 
    \V(a_{\iota_2}' s^{\iota_2}|_i) =  \sum_{\mu_2 = 0}^{K_2} \V(b_{\mu_2}'(\iota_2)|_i) \sum_{j_3=0}^{n-1} \E[e^{\mu_2}|^2_{j_3}] \sum_{j_4=0}^{n-1} \E[s^{\iota_2}|^2_{j_4}] 
\end{cases}
\]
By observing that \[
a_{\iota_1} s^{\iota_1} \cdot a_{\iota_2}' s^{\iota_2}= \left(\sum_{\mu}\sum_{\mu_1+\mu_2 = \mu} b_{\mu_1}(\iota_1)b_{\mu_2}'(\iota_2)  e^{\mu}\right) s^{\iota_1+\iota_2},
\]
and using \Cref{lemma1}, it is possible to write the variance $\V((a_{\iota_1} s^{\iota_1} \cdot a_{\iota_2}' s^{\iota_2})|_i)$ as
\begin{align*}
     &\sum_{\mu} \V\left(\sum_{\mu_1+\mu_2 = \mu} (b_{\mu_1}(\iota_1)b_{\mu_2}'(\iota_2))|_{i} \right) \sum_{j_1=0}^{n-1} \E[e^{\mu}|^2_{j_1}] \sum_{j_2=0}^{n-1} \E[s^{\iota_1+\iota_2}|^2_{j_2}] \\
    &{= n \sum_{\mu} \sum_{\mu_1+\mu_2 = \mu} \V\left( b_{\mu_1}(\iota_1)|_i \right) \V \left(b_{\mu_2}'(\iota_2)|_i\right) \sum_{j_1=0}^{n-1} \E[e^{\mu}|^2_{j_1}] \sum_{j_2=0}^{n-1} \E[s^{\iota_1+\iota_2}|^2_{j_2}] },
\end{align*}
where the second equality follows from the independence of  $b_{\mu_1}(\iota_1), b_{\mu_2}'(\iota_2)$ and from {$\Cov(b_{\mu_1}(\iota)|_{j_1},b_{\mu_2}({\iota})|_{j_2})= 0 $ for $\mu_1 \neq \mu_2 \text{ or } j_1 \neq j_2$}.

Moreover, it can be noted that $n\V(a_{\iota_1} s^{\iota_1}|_i)\V(a'_{\iota_2} s^{\iota_2}|_i)$ can be written as
\begin{align*}
      &n\sum_{\mu}\sum_{\mu_1+\mu_2=\mu} \V(b_{\mu_1}(\iota_1)|_i) \V(b'_{\mu_2} (\iota_2)|_i)  \\ & \hspace{0.2cm} \cdot
      \sum_{j_1=0}^{n-1} \E[e^{\mu_1}|^2_{j_1}] \sum_{j_3=0}^{n-1} \E[e^{\mu_2}|^2_{j_3}] \sum_{j_2=0}^{n-1} \E[s^{\iota_1}|^2_{j_2}] \sum_{j_4=0}^{n-1} \E[s^{\iota_2}|^2_{j_4}],
\end{align*}

By \eqref{appF}, we express $\V((a_{\iota_1} s^{\iota_1} \cdot a_{\iota_2}' s^{\iota_2})|_i)$ in terms of $\V(a_{\iota_1} s^{\iota_1}|_i)\V(a'_{\iota_2} s^{\iota_2}|_i)$. In fact, $\V((a_{\iota_1} s^{\iota_1} \cdot a_{\iota_2}' s^{\iota_2})|_i)$ can be written as
\begin{equation*}
     n \sum_{\mu} \sum_{\mu_1+\mu_2 = \mu} \V\left( b_{\mu_1}(\iota_1)|_i \right) \V \left(b'_{\mu_2}(\iota_2)|_i\right)  \sum_{j=0}^{n-1} \E[e^{\mu_1+\mu_2}|^2_{j}] \sum_{j'=0}^{n-1} \E[s^{\iota_1+\iota_2}|^2_{j'}].
\end{equation*}
Thus, by leveraging the properties of the correction functions
\begin{align*}
 n F_s(\iota_1,\iota_2) &\sum_{\mu} \sum_{\mu_1+\mu_2 = \mu} \V\left( b_{\mu_1}(\iota_1)|_i \right) \V \left(b'_{\mu_2}(\iota_2)|_i\right)F_e(\mu_1,\mu_2) \cdot \\  \cdot &\sum_{j_1=0}^{n-1} \E[e^{\mu_1}|^2_{j_1}] \sum_{j_2=0}^{n-1} \E[e^{\mu_2}|^2_{j_2}] \sum_{j_3=0}^{n-1} \E[s^{\iota_1}|^2_{j_3}]\sum_{j_4=0}^{n-1} \E[s^{\iota_2}|^2_{j_4}] \\
 \approx \quad & n F_s(\iota_1,\iota_2) \sum_{\mu} \sum_{\mu_1+\mu_2 = \mu} \V\left( b_{\mu_1}(\iota_1)|_i \right) \V \left(b'_{\mu_2}(\iota_2)|_i\right)F_e(\mu_1,\mu_2) \cdot \\  \cdot &\sum_{j_1=0}^{n-1} \E[e^{\mu_1}|^2_{j_1}] \sum_{j_2=0}^{n-1} \E[e^{\mu_2}|^2_{j_2}] \sum_{j_3=0}^{n-1} \E[s^{\iota_1}|^2_{j_3}]\sum_{j_4=0}^{n-1} \E[s^{\iota_2}|^2_{j_4}].
\end{align*}
Then, by exploiting the monotonicity of \( F_e \) and \( F_s \), it is possible to derive an upper bound given by
\begin{align*}
     n F_s(\iota_1,\iota_2) F_e(K_1,K_2)& \sum_{\mu} \sum_{\mu_1+\mu_2 = \mu} \V\left( b_{\mu_1}(\iota_1)|_i \right) \V \left(b'_{\mu_2}(\iota_2)|_i\right) \cdot \\ \quad \cdot &\sum_{j_1=0}^{n-1} \E[e^{\mu_1}|^2_{j_1}] \sum_{j_2=0}^{n-1} \E[e^{\mu_2}|^2_{j_2}] \sum_{j_3=0}^{n-1} \E[s^{\iota_1}|^2_{j_3}]\sum_{j_4=0}^{n-1} \E[s^{\iota_2}|^2_{j_4}]   ,
\end{align*}
where $K_1,K_2$ represent the highest power of $e$ appearing in  $a_{\iota_1}, a'_{\iota_2}$, respectively. \\ It is straightforward to verify that this concludes our proof, yielding \begin{equation*}
   \V((a_{\iota_1} s^{\iota_1} \cdot a_{\iota_2}' s^{\iota_2})|_i) '\le n  \V(a_{\iota_1} s^{\iota_1}|_i)\V(a'_{\iota_2} s^{\iota_2}|_i)F_s(\iota_1,\iota_2) F_e(K_1,K_2).
\end{equation*}
\qed
\end{proof}
\section{The Shape of Noise: Gaussian Distributions in BGV}\label{sec:BGV_Gaussian}

 In order to provide accurate bounds for the selection of initial parameters, it is necessary to first discuss the distribution of the error coefficients. Our analysis relies on the assumption that these coefficients follow a normal distribution. Under this assumption, tight average-case bounds can be derived, facilitating efficient and reliable parameter selection.
 
However, assuming a Gaussian distribution when this does not hold may result in a significant underestimation of the error bounds, causing computation failures, as highlighted in \cite{cryptoeprint:2025/1036}. It is worth noting that \cite{cryptoeprint:2025/1036} claims that variance-based methods offer no theoretical guarantee on the failure probability. This statement can be misleading, as variance-based analyses remain valid once an appropriate bound on the actual distribution is established. 

Furthermore, we emphasize that the results proven in \cite{cryptoeprint:2025/1036} concern the heavy-tailed nature of the error coefficient distributions in the absence of modulus switching, which is not practical, since BGV is always used with modulus switching technique. As noted in the article itself, applying modulus switching ensures that the dominant terms are effectively distributed according to Gaussian distributions. Therefore, assuming Gaussianity remains reasonable when analyzing BGV parameters, in line with other works such as \cite{costache2023optimisations,murphycentral}.  
Indeed,  as the authors of~\cite{cryptoeprint:2025/1036} themselves note, no evidence of decryption failures has been observed in practical BGV usage scenarios, nor in those of CKKS~\cite{DBLP:journals/cic/BossuatCMNT25}.

This section aims to rigorously justify why the Gaussian assumption for the error coefficients is reasonable and to identify the conditions on the choice of initial parameters that are required to ensure it.
In particular, we will show that the values of the primes $p_\ell$ must satisfy a certain lower bound. This requirement is reasonable: the primes used in the most common libraries already meet this criterion. Moreover, while imposing a minimum size on these primes, our result still allows for improvements and a reduction in the overall modulus size.

\subsection{Gaussian Distribution}

In this section, we provide both theoretical arguments and empirical evidence that the use of modulus switching ensures that the error coefficients are, in practice, Gaussian distributed.\\
We show that the coefficients of the critical quantity can be reasonably assumed to follow a Gaussian distribution at three distinct stages: immediately after the initial encryption, after modulus switching, and after the multiplication of two ciphertexts that have both undergone modulus switching beforehand. Specifically, we provide a formal proof for the first two cases, while for the latter we verify this behavior experimentally. 
\paragraph{Gaussianity of the Fresh Error.}
The critical quantity associated with a fresh ciphertext is given by \begin{equation*}
    m + t(e \cdot u + e_1 \cdot s + e_0),
\end{equation*}
where $e, e_0, e_1 \sim \mathcal{DG}_q(\sigma^2)$ and $m, s, u \sim \mathcal{U}_3$. We are thus interested in the distribution of the coefficients \begin{equation*}
     m|_i + t( \sum_{j=0}^{n-1}e|_j u|_{i-j} + \sum_{j=0}^{n-1}e_1|_js|_{i-j} + e_0|_{i} ), 
\end{equation*}
where $m|_i \longleftarrow \mathcal{U}_t$, $te_0|_{i} \longleftarrow \mathcal{DG}_q(t^2\sigma^2)$ and $t\sum_{j=0}^{n-1}e|_j u|_{i-j}, t\sum_{j=0}^{n-1}e_1|_js|_{i-j} \longleftarrow \mathcal{DG}_q(t^2nV_eV_s)$.
Since the sum of independent Gaussian random variables is still Gaussian with variance equal to the sum of individual variances \cite{feller1950introduction}, the error term can be written as
\begin{equation*}
    m|_i + K|_i,
\end{equation*}
where $m|_i \longleftarrow \mathcal{U}_t$ and 
$K|_i \longleftarrow \mathcal{DG}_q(t^2V_e(2nV_s+1))$.
Therefore, the Gaussianity of the final critical quantity follows from the fact that the first term is negligible compared to the second.
In \Cref{gaussian_fresh}, we present the results for circuits of multiplicative depth 3 (\Cref{fresh_M3_s50_t65537}) and 6 (\Cref{fresh_M6_50k}), focusing on the distribution of the first coefficient of the fresh error.
We used OpenFHE library \cite{OpenFHE} to generate 50,000 error samples, and we analyzed their coefficients using the Python \textbf{fitter} package\footnote{\url{https://fitter.readthedocs.io/en/latest/}}. The code used to generate these samples and to perform the estimates in the following sections is publicly available\footnote{\url{https://github.com/nadirmur/openFHE}}.
As the figures show, the distribution of the fresh error coefficient is well approximated by a Gaussian. In particular, the Kolmogorov–Smirnov test produces a $p$-value $> 0.05$ (ks\textsubscript{pval} in the caption) \cite{sheskin2003handbook}, the Anderson-Darling test statistic (ad\textsubscript{stat}) is below the critical value (ad\textsubscript{crit}) \cite{anderson1954test,stephens2017tests} at the $15\%$ significance level, and the kurtosis is approximately $3$.
The parameters used are: $t = 65537$, $n = 2^{13}$, ciphertext modulus $q$ chosen by the library to ensure at least 128-bit security, $\chi_s = \chi_u = \U_3$, and $\chi_e = \mathcal{DG}( \sigma^2)$ with $\sigma = 3.19$. We employed hybrid key switching and the \texttt{HPSPOVERQ} multiplication method \cite{OpenFHE}.

\begin{figure}[htb]
    \centering
    \begin{subfigure}{.47\textwidth}
         \centering
        \includegraphics[width=\textwidth]{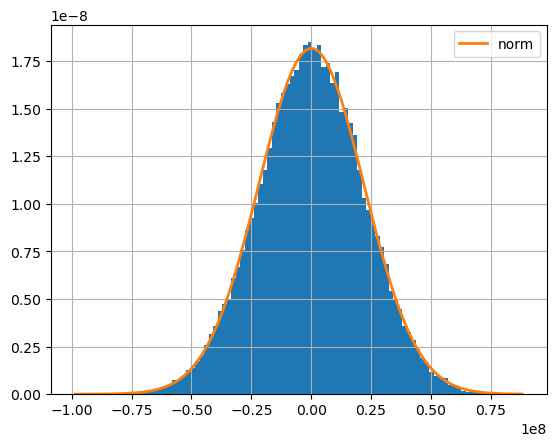}
        \caption{ad\textsubscript{stat}~0.517 $<$ ad\textsubscript{crit} 0.576 ($\alpha=0.15$), ks\textsubscript{pval} 0.896 and kurtosis 2.999. Circuit of multiplicative depth 3. }
        \label{fresh_M3_s50_t65537}
    \end{subfigure}
    \hspace{.2cm}
    \begin{subfigure}{.47\textwidth}
         \centering
        \includegraphics[width=\textwidth]{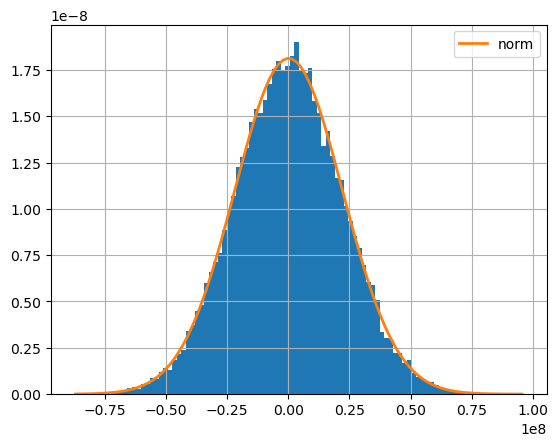}\caption{ad\textsubscript{stat}~0.304 $<$ ad\textsubscript{crit} 0.576 ($\alpha=0.15$), ks\textsubscript{pval} 0.940 and kurtosis 3.00. Circuit of multiplicative depth 6.}
        \label{fresh_M6_50k}
    \end{subfigure}
    \caption{Distribution of the first coefficient of the fresh error}
     \label{gaussian_fresh}
    \medskip
\end{figure}
\paragraph{Gaussianity after Modulus Switching.}
Given $\fc = (\vc, q_\ell,\nu)$ an extended ciphertext, we recall that the variance of the error coefficients after modulo switching to the target modulo $q'_\ell$, where $q_\ell = p_\ell q'_\ell$, is 
\begin{equation}\label{eq.V.ms2}
    V(\nums|_i)=\tfrac{1}{p_\ell^2} (\V(\nu|_i) + V_{\delta_0} + nV_{\delta_1}V_s) \approx \tfrac{1}{p_\ell^2} (\V(\nu|_i) + nV_{\delta_1}V_s).
\end{equation}
as $\V({\delta_0}/{p_\ell}|_i) = t^2/12 \ll \V({\delta_1}s/{p_\ell}|_i) = nt^2V_s/12$.

In our analysis, we will assume that the dominant component in \eqref{eq.V.ms2} is the second term. This assumption is quite common in the study of BGV, as can be seen in \cite{costache2023optimisations}. We will later clarify how our choice of moduli is specifically designed to ensure that this condition is satisfied.

It is precisely this constraint that allows us to conclude the Gaussianity of the error coefficients after modulus switching. By making the first term negligible, the error after modulus switching reduces, according to the previous notation, to
${\delta_1}/{p_\ell}s, $ with ${\delta_1}/{p_\ell}\longleftarrow\mathcal{U}_t$.
Thus, we obtain
\begin{equation*}
    \frac{\delta_1}{p_\ell}s|_i = \sum_{j=0}^{n-1} \xi (i,j) \hspace{0.1cm} \frac{\delta_1}{p_\ell}|_j \hspace{0.1cm} s|_{i-j}.
\end{equation*}
Therefore, we can observe that each coefficient is given by the sum of i.i.d. random variables with finite mean and variance. Since $n$ is sufficiently large, by applying the Central Limit Theorem, we can conclude that \begin{equation*}
     \frac{\delta_1}{p_\ell}s\longleftarrow\mathcal{DG}_q(n\frac{t^2}{12}V_s).
\end{equation*}



\paragraph{Gaussianity after a multiplication.} 

In \Cref{gaussian_ms}, we present the results for circuits of multiplicative depth 3 (\Cref{mul_M3_s50_t65537}) and 4 (\Cref{mul_M4_s50_t65537}), focusing on the distribution of the first coefficient of the error just after the last multiplication. Moreover,  \Cref{3mult_ms_L6,5mult_ms_L6} show the multiplication error after the third and fifth multiplications, respectively, for a circuit of depth 6.
We used OpenFHE library \cite{OpenFHE} (with the same setting as before) to generate 50,000 error samples and analyzed their coefficients using the Python \textbf{fitter} package. As the figures show, the distribution of the error coefficients is well approximated by a Gaussian.

\begin{figure}[htb]
    \centering
    \begin{subfigure}{.47\textwidth}
        \centering
        \includegraphics[width=\textwidth]{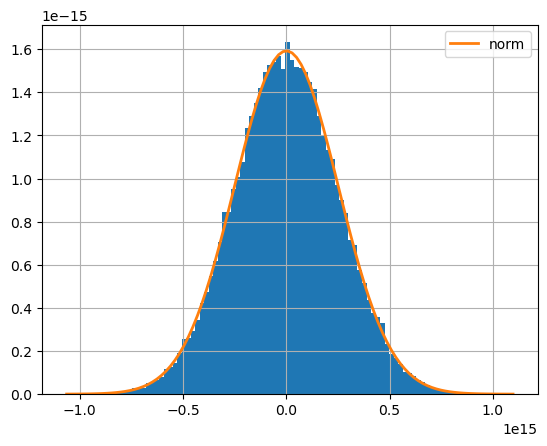}
        \caption{ad\textsubscript{stat}~$0.267<$ ad\textsubscript{crit} $0.576$~($\alpha=0.15$), ks\textsubscript{pval} $0.867$ and kurtosis $2.987$. After the third multiplication.}
        \label{mul_M3_s50_t65537}
    \end{subfigure}
    \hspace{.2cm}
    \begin{subfigure}{.47\textwidth}
        \centering
        \includegraphics[width=\textwidth]{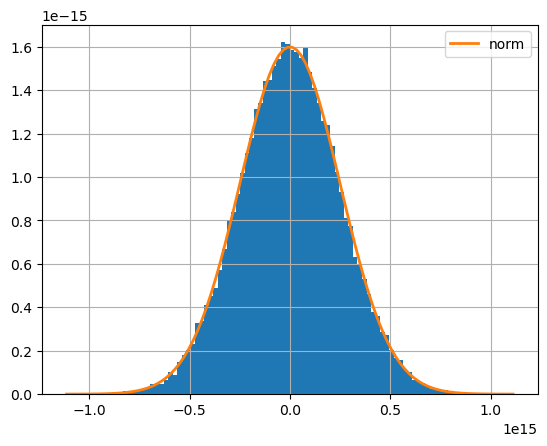}
        \caption{ad\textsubscript{stat}~$0.306<$ ad\textsubscript{crit} $0.576$~($\alpha=0.15$), ks\textsubscript{pval} $0.937$ and kurtosis $3.03$. After the fourth multiplication.}
        \label{mul_M4_s50_t65537}
    \end{subfigure}
%
    \begin{subfigure}{.47\textwidth}
        \centering
        \includegraphics[width=\textwidth]{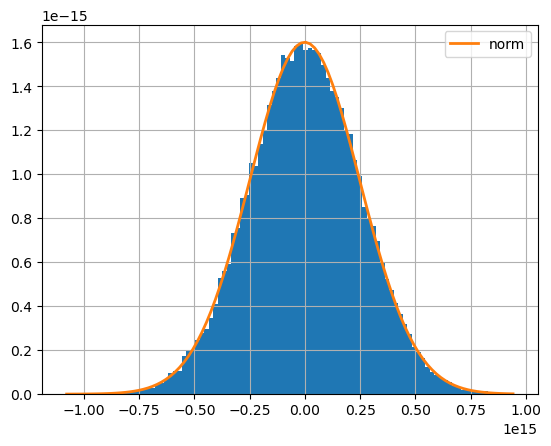}
        \caption{ad\textsubscript{stat}~0.233 $<$ ad\textsubscript{crit} $0.576$ ($\alpha=0.15$), ks\textsubscript{pval} $0.985$ and kurtosis $2.980$. After the 3rd multiplication in a circuit of multiplicative depth 6.}
        \label{3mult_ms_L6}
    \end{subfigure}
    \hspace{.2cm}
    \begin{subfigure}{.47\textwidth}
        \centering
        \includegraphics[width=\textwidth]{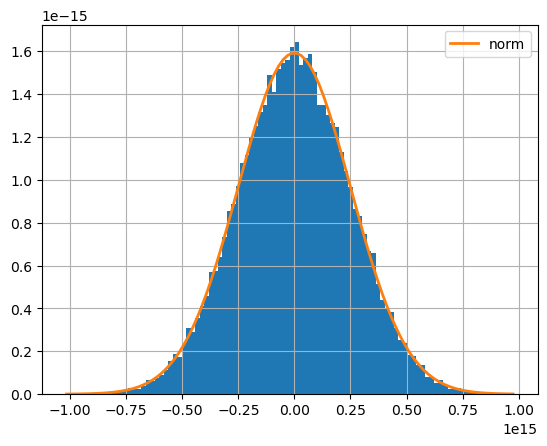}
        \caption{ad\textsubscript{stat}~0.215 $<$ ad\textsubscript{crit} $0.576$ ($\alpha=0.15$), ks\textsubscript{pval} $0.842$ and kurtosis $2.985$. After the 5th multiplication in a circuit of multiplicative depth 6.}
        \label{5mult_ms_L6}
    \end{subfigure}
    \caption{Distribution of the first coefficient just after a multiplication.\\}
    \label{gaussian_ms}
\end{figure}

The fact that the noise coefficients follow a Gaussian distribution is particularly advantageous, as it allows us to bound the maximum absolute value of the noise coefficients with high probability simply by controlling their variance \( V \).
This implies that, to satisfy the correctness condition, a noise vector \( \nu \) must satisfy 
$\| \nu \|_{\infty} < {q}/{2}.$
We can use \Cref{erf} to bound the failure probability
\begin{equation*}
\mathbb{P} \left( \|\nu\|_{\infty} > \frac{q}{2} \right) \leq n \left(1 - \text{erf}\left(\frac{q}{2\sqrt{2V}}\right) \right),
\end{equation*}
where \( V \) denotes the estimated variance of the noise coefficients.

To express this bound more conveniently, we introduce a \textit{security parameter} \( D \) such that
$D \leq {q}/{2\sqrt{2V}}$,
from which it follows, using the monotonicity of the error function, that
\begin{equation*}
\mathbb{P}\left( \|\nu\|_{\infty} > \frac{q}{2} \right) \leq n \left(1 - \text{erf}(D) \right).
\end{equation*}
Thus, by appropriately choosing the security parameter \( D \), we can ensure that the probability of decryption failure remains negligibly small. For instance, setting \( D = 6 \) and \( n = 2^{13} \) yields a failure probability of approximately \( 2^{-42} \).  For practical applications, $D=8$ should be preferred, resulting in a probability of approximately $2^{-83}$, when the ring dimension is $n=2^{13}$.

Consequently, the ciphertext modulus \( q \) can be selected as
\begin{equation}\label{ourqProb}
q \geq 2D \sqrt{2V}.
\end{equation}

It is worth noting that the bounds derived in our analysis provide insight into the \textit{minimum} ciphertext modulus \( q \) required to guarantee the correctness of the scheme, a key factor for optimizing performance and efficiency. 

\subsection{On the Role of Modulus Switching in Preserving Gaussianity}






In this section, we highlight the crucial role of modulus switching in the studied circuits. This mechanism is essential not only to control the noise growth induced by homomorphic operations, but also to \textit{shape} the error distribution. 
Our experiments confirm this behavior: when modulus switching is omitted, the Gaussianity of the coefficients could no longer hold, making it necessary to explicitly characterize their distribution. 
Specifically, \Cref{mul_M3_no_ms} shows the error distribution after three multiplications without modulus switching for $n=2^{13}$, where the tails clearly deviate from Gaussian behavior. 

\begin{figure}[H]
    \centering
    \begin{subfigure}{.46\textwidth}
         \centering
        \includegraphics[width=\textwidth]{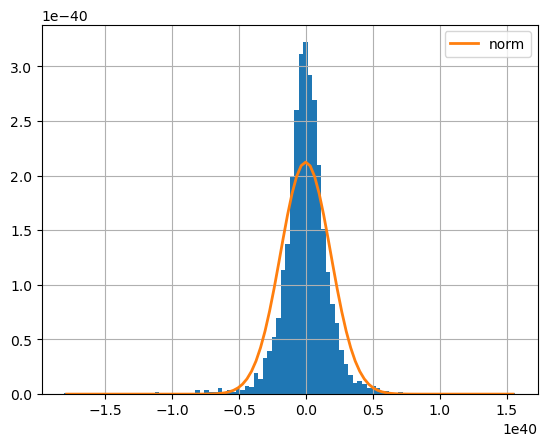}
        \caption{asymmetry $-0.5240$ and kurtosis $14.3769$}
    \end{subfigure}
    \hspace{.2cm}
    \begin{subfigure}{.49\textwidth}
         \centering
        \includegraphics[width=\textwidth]{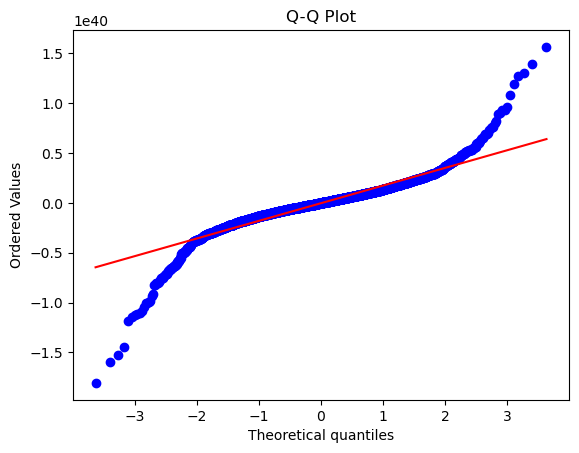}
        \caption{when the distribution is Gaussian, the points lie along the diagonal.}
    \end{subfigure}
    \caption{Error distribution after 3 multiplications without modulus switching.}
     \label{mul_M3_no_ms}
    \medskip
\end{figure}

It is worth noting that this deviation can be observed even without plotting the distribution. Indeed, when the Gaussian bound introduced in the previous section is applied to select the ciphertext modulus \(q\) according to \Cref{ourqProb}, the resulting parameters quickly lead to an unexpectedly high failure rate. For instance, when targeting a decryption failure probability of approximately \(2^{-83}\) (as obtained by setting \(D=8\) in our bound), experiments without modulus switching exhibit failure rates of around 3\%. This discrepancy provides strong evidence that the Gaussian assumption no longer holds and that the tails of the distribution are significantly heavier, in agreement with \cite{cryptoeprint:2025/1036}.

The validity of our approach is therefore directly linked to the use of modulus switching. As previously shown, the error coefficients after encryption, modulus switching, and multiplication preserve their Gaussian distribution. For this property to hold, as highlighted in the previous sections, it is necessary to ensure that the dominant terms after modulus switching and after multiplication are, respectively, $\delta_1s/p_\ell$ and  $(\delta_1/p_\ell \cdot \delta_1'/p_\ell)s^2$. For this reason, in our moduli selection, we explicitly require that 
\begin{equation}\label{gaussian.condition}
   \frac{V_{\ell-1}}{p_{L-\ell}^2} < \alpha \Vms.
\end{equation}
for every level $\ell$, with $\alpha = \sfrac{1}{100}$.
This condition imposes a lower bound on the moduli $p_\ell$, which must be selected accordingly. Nevertheless, it does not compromise the modulus-size reduction achieved by our approach compared to the state-of-the-art. In fact, the primes generated by most existing libraries already satisfy this requirement.

\ec

\section{Variance Estimation in Circuits with Fixed Depth}\label{sec:BGV_circuit}

This section builds on the results of Section~\ref{sec:BGVaverage_teo} and \ref{sec:BGV_Gaussian} to propose a method for tracking the error growth in circuits with a fixed number of operations. As previously noted, obtaining sufficiently tight bounds with respect to the experimental variance of the error coefficients is crucial for identifying initial scheme parameters that simultaneously improve both performance and security.  

We analyze the circuit depicted in \Cref{fig:circuit_mod} in which pairs of ciphertexts are progressively multiplied. Note that we focus on circuits that involve homomorphic multiplication since it is the most complex and, therefore, the most significant operation to study. It is worth noting that the circuit described in this work closely resembles the one adopted by default in OpenFHE \cite{OpenFHE}, where multiplications are always preceded by a modulus switching operation.

%

\begin{figure}[htb] 
  \centering 
  \includegraphics[width=0.9\textwidth]{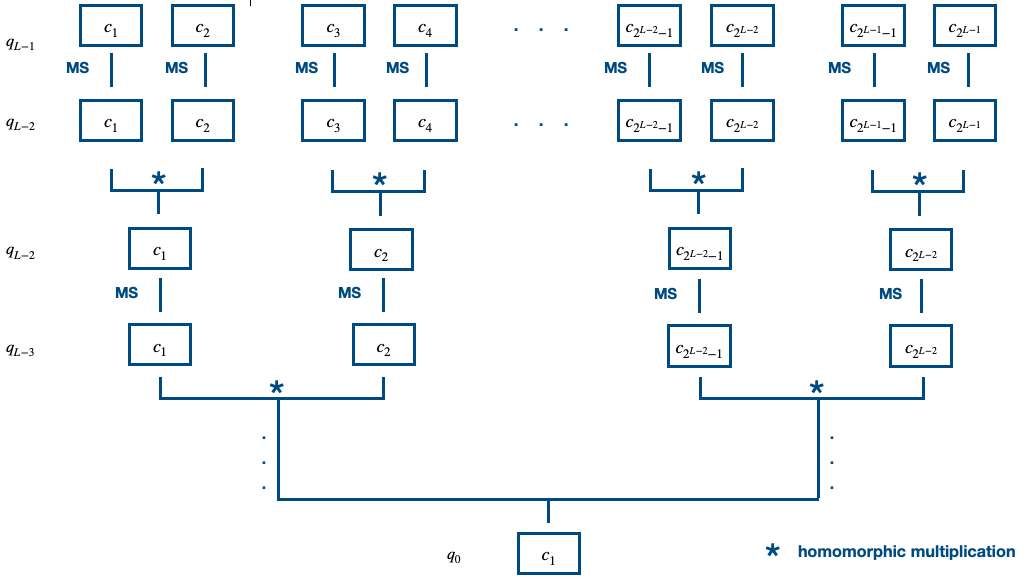} 
  \caption{Reference circuit\\} 
  \label{fig:circuit_mod} 
\end{figure}

Let $L$ denote the number of levels in the circuit, and $\bm{c}_1,\dots,\bm{c}_{2^{L-1}}$ the initial fresh ciphertexts, generated by encrypting $2^{L-1}$ independent and randomly generated messages using the same key. 

 We consider a model where, at each level, homomorphic multiplication of pairs of ciphertexts is carried out. At the beginning of each level $\ell \in \{1,\dots,L-1\}$ the input ciphertexts are switched to a smaller modulus, to maintain the error almost constant throughout the process. We recall that, $q_\ell$'s are the ciphertext moduli of each level $L-1-\ell$, defined as $q_\ell=\prod_{j=0}^\ell p_j$,
where $p_j$ are primes such that $\gcd(p_i,p_j)=1$ for $i \ne j$. 

We denote by  $\nu_{\ell}^{\mathsf{ms}}$ the critical quantity obtained after performing the modulus switching at the beginning of level $\ell$, and by $\nu_\ell$ the critical quantity at the end of level $\ell$, namely just after the multiplication (and key switching technique). Similarly, $V_\ell^{\mathsf{ms}}$ and $V_\ell$ denote the variance of the coefficients of these two critical quantities, respectively. Finally, for each level $\ell$ we will use the notation $\Vms$ to denote $\V(\delta_1s/p_\ell|_i)$. 

The process represented by the circuit can be summarized as follows: \begin{itemize}
    \item At level 0, all the fresh ciphertexts are defined in $\R_{q_{L-1}}^2$.
    \item At level 1, modulo switch to $q_{L-2}$ is applied to each of the initial fresh ciphertexts. After that, the first homomorphic multiplication is performed in $\R_{q_{L-2}}$, yielding $2^{L-2}$ resulting ciphertexts in $\R_{q_{L-2}}^3$. Finally, these ciphertexts are relinearized to obtain the equivalent ones in $\R_{q_{L-2}}^2$, which are the inputs of the next level.  
    \item This process is repeated until level $L-1$, where a final ciphertext, which is the output of the circuit, is returned. 
\end{itemize} 
It should be noted that, in the circuit represented in Fig. \ref{fig:circuit_mod}, the relinearization step is not presented.
In fact, we decided to omit its contribution from our error analysis, since it is negligible compared to the impact of modulus switch and multiplication. A detailed justification for this choice is provided in Appendix~\ref{KS.negligibility}.

\paragraph{Level $0$.}
Let $\bm{c}\in \R_{q_{L-1}}^2$ be a fresh ciphertext, obtained encrypting a random message $m \in \R_t$.  
Let \( \nuenc = a_0 + a_1s \) be the critical quantity associated with $\bm{c}$ and let $V_0$ be the variance of $\nuenc$. Then, we have
\begin{equation*}
\begin{cases}
     V(a_0|_i) = t^2(\frac{1}{12}+ V_e + nV_eV_u) \\ 
     V(a_1s|_i) = t^2nV_eV_s.
\end{cases}
\end{equation*}
Therefore, it is possible to assume 
\begin{equation}\label{V0}
    \V(a_0|_i) \approx \V((a_1s)|_i) \approx \frac{V_0}{2}.
\end{equation}  
It should be noted that the assumption in \eqref{V0} applies to all initial $2^{L-1}$ ciphertexts.  In fact, the estimate of variance $V_\ell$ of the coefficients of the critical quantity at the end of level $\ell$,  is the same for all ciphertexts belonging to the same level.
\paragraph{Level $1$.}
At the beginning of this level, all ciphertexts are subject to a modulus switching from $q_{L-1}$ to $q_{L-2}$. Therefore, the error and the modulus of $\bm{c}$ are rescaled by a factor 
${q_{L-2}}/{q_{L-1}} = {1}/{p_{L-1}}$.
Thus, the critical quantity associated with the resulting ciphertext is  \begin{equation}\label{noiseMod}
   \nu_1^{\mathsf{ms}} =  \frac{\nuenc + \delta_0 + \delta_1s}{p_{L-1}}, \quad \delta_i = t[-c_it^{-1}]_{p_{L-1}}
\end{equation}
Since $\delta_i/p_{L-1} \leftarrow \mathcal{U}_t$, it follows that 
    $\V(\delta_0/p_{L-1}|_i) = t^2/12 $ and $
    \V(\delta_1s/p_{L-1}|_i) ={t^2nV_s}/{12}$.  In BGV, $n$ is typically larger than $2^{12}$, making the contribution of $\delta_0/p_{L-1}$ negligible. 


The critical quantity in \eqref{noiseMod} can be written as
\begin{equation*}
    \nu_1^\mathsf{ms} \approx  \frac{a_0}{p_{L-1}}  + \left( \frac{a_1}{p_{L-1}} + \frac{\delta_1}{p_{L-1}} \right) s,
\end{equation*}
By \Cref{V0} we have
$$\V\left(\frac{a_0}{p_{L-1}}\middle|_i\right) =  \V\left(\frac{a_1}{p_{L-1}}s\middle|_i\right) =  \frac{V_0}{2 p_{L-1}^2}.$$
Therefore, the variance after the modulus switch can be written as 
\begin{equation*}
    V_1^\mathsf{ms} = \frac{V_0}{p_{L-1}^2} + \Vms.
\end{equation*}
As discussed previously, to ensure that the error follows a Gaussian distribution, we rely on \Cref{gaussian.condition}, i.e., we can assume ${V_0}/{p_{L-1}^2} < \Vms/100$.

In order to estimate the variance of the critical quantity resulting from multiplication, we rely on \Cref{teorema1}. 
In particular, we avoid explicitly writing the terms \( b_{\mu}(\iota) \) in the expression \( a_{\iota} = \sum_{\mu = 0}^{K} b_{\mu}(\iota) e^{\mu} \) to keep the notation manageable. However, it is crucial to keep track of the highest power of \( e \), which is \( K \), appearing in the coefficients \( a_{\iota} \) involved in the multiplications, since the correction applied by the function \( F \) depends on this value.

Precisely, given the fresh ciphertexts, it can be seen that for $a_0$ the associated $K$ is $1$, while for $a_1$ we have $K=0$. Instead, for the terms $ \delta_0/p_\ell, \delta_1/p_\ell$ the case is quite different. In fact, these values can be assumed to be randomly distributed over $\mathcal{R}_t$ 
, for every level of the circuit. This makes the analysis slightly more complicated, since, in order to derive tight estimates, it will be necessary to distinguish the contribution of ${a_0,a_1}$ from the one of $\delta_0,\delta_1$.

Now, assume that $\bm{c}, \bm{c}' \in \R_{q_{L-2}}^2$ are two ciphertexts at level 1, after modulus switching has been carried out. Their associated critical quantity is given, respectively, by \begin{equation*}
    \begin{cases}
        \nu_1^\mathsf{ms} = \frac{1}{p_{L-1}}(a_0 + a_1s) +  \frac{\delta_1}{p_{L-1}}s \\
         {\nu}_1^{\mathsf{ms}'}= \frac{1}{p_{L-1}}(a'_0 + a'_1s) +  \frac{\delta'_1}{p_{L-1}}s \\
    \end{cases}
\end{equation*}
Therefore, the critical quantity obtained after their multiplication is of the form  \begin{align*}
  \nu_{1} &=\frac{1}{p_{L-1}^2}a_0a'_0 +\frac{1}{p^2_{L-1}}\left( a'_0 {a_1} + a_0 {a'_1}\right)s + \frac{1}{p_{L-1}}(a'_0\frac{\delta_1}{p_{L-1}} + a_0 \frac{\delta'_1}{p_{L-1}} )  s \hspace{0.15cm} \\ &\hspace{0.4cm}+  \frac{1}{p^2_{L-1}}a_1a'_1s^2 + \frac{1}{p_{L-1}}(a_1\frac{\delta'_1}{p_{L-1}} +a'_1\frac{\delta_1}{p_{L-1}})s^2+  \frac{\delta'_1}{p_{L-1}}\frac{\delta_1}{p_{L-1}}s^2.
\end{align*}
So, we are now able to provide an estimate of the variance $V_1$, applying \Cref{teorema1}, as follows 
\begin{align*}
    V_1  &\le
     \frac{n}{p_{L-1}^4}V(a_0|_i)V(a'_0|_i)F_e(1,1) + \frac{n}{p_{L-1}^4}V(a'_0|_i)V(a_1s|_i) + \frac{n}{p_{L-1}^4}V(a_0|_i)V(a'_1s|_i) \\
     &+ \frac{n}{p_{L-1}^2}V(a'_0|_i)V(\frac{\delta_1}{p_{L-1}}s|_i) + \frac{n}{p_{L-1}^2}V(a_0|_i)V(\frac{\delta'_1}{p_{L-1}}s|_i)\\
      &+ \frac{n}{p_{L-1}^4}V(a_1s|_i)V(a'_1s|_i)F_s(1,1)  + \frac{n}{p_{L-1}^2}V(a_1s|_i)V(\frac{\delta'_1}{p_{L-1}}s|_i)F_s(1,1)  \\
    &+ \frac{n}{p_{L-1}^2}V(a'_1s|_i)V(\frac{\delta_1}{p_{L-1}}s|_i)F_s(1,1)+ {n}V(\frac{\delta_1}{p_{L-1}}s|_i)V(\frac{\delta'_1}{p_{L-1}}s|_i)F_s(1,1).   
\end{align*}
which, recalling that $\Var(a_{\iota}|_i)=\Var(a'_{\iota}|_i) = {V_0}/{2}$ and that $V(\frac{\delta_1}{p_{L-1}}s|_i)=V(\frac{\delta'_1}{p_{L-1}}s|_i)=\Vms$ can be rewritten as
\begin{align*}
    V_1 &\le \frac{n}{p_{L-1}^4}\frac{V_0^2}{4}F_e(1,1) + \frac{2n}{p_{L-1}^4}\frac{V_0^2}{4} + \frac{2n}{p_{L-1}^2}\Vms \frac{V_0}{2} \\ &+ \frac{n}{p_{L-1}^4}\frac{V_0^2}{4}F_s(1,1) + \frac{2n}{p_{L-1}^2}\Vms \frac{V_0}{2} \bcR F_s(1,1)\ec + nF_s(1,1)\Vms^2.
\end{align*}
By observing that the condition in \Cref{gaussian.condition} holds, our bound on $V_1$ is derived as follows \begin{equation*}
    V_1 \approx \left(\frac{3}{2\cdot 100^2} + \frac{3}{100} + 2\right)n\Vms^2 \approx 2n\Vms^2,
\end{equation*}
and we can simplify $\nu_1 = \delta_1 \delta'_1s^2/p_{L-1}$.

We now present the generic level $\ell$. We begin by illustrating the case $\ell = 2$. As we shall see, the input to level 2 corresponds to that of a generic level $\ell \ge 2$, and it produces the same type of output, thus making the analysis carried out for the second level applicable to all subsequent levels.

\paragraph{Level ${2}$.}
At the beginning of level $2$ modulo switch is performed, yielding a critical quantity of the form 
$\nu_2^\mathsf{ms} = {\nu_1}/{p_{L-2}} + {\delta_1}s/{p_{L-2}}$ \,,
with coefficient variance,
\begin{equation*}
    V_2^\mathsf{ms} = \frac{V_1}{p_{L-2}^2} + \Vms \le \left(\frac{1}{100}+1\right)\Vms \,.
\end{equation*} 
Now, the product of two critical quantities $\nu_1^{\mathrm{ms}}$ can be expanded into four terms: the first is the product of two $\nu_1 / p_{L-2}$ terms, the second and third are the products of a  $\nu_1 / p_{L-2}$  term with a  $\nu_1^{\mathrm{ms}}$ term and the fourth is the product of two $\nu_1^{\mathrm{ms}}$ terms. Applying our theorem, the variance after the multiplication is then given by
\begin{equation*}
    V_2 = n\left(\frac{V_1^2}{p_{L-2}^4}F_s(2,2) +2\frac{V_1}{p_{L-2}^2}\Vms F_s(2,1) + 2\Vms^2  \right) ,
\end{equation*}
which can be bounded, for the conditions over $p_{L-2}$, as \begin{equation*}
    V_2 \le \left(\frac{6}{100^2} + \frac{6}{100} + 2\right)n\Vms^2 \approx 2n\Vms^2 \,.
\end{equation*}
Again, the critical quantity at the end of level 2 can be simplified as $\nu_2 = \delta_1 \delta'_1s^2/p_{L-2}$. From this observation, we can deduce the variance estimates for each level $\ell$.

\paragraph{Level $\ell$.}
For any level $\ell \ge 2$, the variance after modulus switching $V_{\ell}^\mathsf{ms}$ and the variance after multiplication $V_{\ell}$ can be approximated as \begin{align}\label{ourVariance}
    &V_{\ell}^\mathsf{ms} \approx \, \frac{101}{100} \Vms \notag\\
    &V_{\ell} \, \, \,\,\approx \,(2 + \epsilon)n\Vms^2 , \quad\epsilon= \frac{6}{100^2} + \frac{6}{100}
\end{align}

\Cref{variance_mp_comparison} presents a comparison of the estimated variances of the error coefficients obtained using our approach, denoted as \textit{our} and the corresponding experimental values, denoted as \textit{exp}. 
 The estimates were obtained using the OpenFHE library, setting the parameters according to the required multiplicative depth, with $t = 65537$. For the experimental values, 50000 samples were computed for $n=2^{13}$, 8000 for $n=2^{14}$ and 5000 for $n=2^{15}$. To enhance readability, variances are presented in terms of their base-2 logarithms.
\begin{table}[htb]
    \centering
    \footnotesize
    \begin{tabular}{ccccccccc}
        \toprule
        & \multicolumn{2}{c}{Encryption} & \multicolumn{2}{c}{ Modulo Switch} & \multicolumn{2}{c}{1 Multiplication} & \multicolumn{2}{c}{6 Multiplications} \\
        \cmidrule(lr){2-3} \cmidrule(lr){4-5} \cmidrule(lr){6-7} \cmidrule(lr){8-9}
        $n$ & \textit{our} & \textit{exp} & \textit{our} & \textit{exp} & \textit{our} & \textit{exp} & \textit{our} & \textit{exp} \\
        \midrule
        $2^{13}$ & 48.76 & 48.76 & 40.84 & 40.83 & 95.68 & 95.65 & 95.71 & 95.25 \\
        $2^{14}$ &  49.76  & 49.76 & 41.84 & 41.79 & 98.68 & 97.62 & 98.71 & 97.63 \\
        $2^{15}$ & 50.76 & 50.72 & 42.84 & 42.82 & 101.68  & 101.62 & 101.71 & 101.59 \\
        \bottomrule
    \end{tabular}
    \vspace{.2cm}
    \caption{Encryption, modulo switch, and multiplication of fresh ciphertexts. The last two columns report the case of six consecutive multiplications as in the reference circuit.}
    \label{variance_mp_comparison}
\end{table}

\begin{obs}
    At first glance, it might seem that, in order to correctly estimate the variance of the multiplication, it is sufficient to apply a factor of 2, arising from the correction of the product between two $s$-terms, thus making it unclear why the dependence on the public key should be considered. It is crucial to emphasize that these estimates are strictly related to a circuit model that we are specifically constructing to guarantee the Gaussianity of the error coefficients according to \Cref{gaussian.condition}. However, if one wanted to accurately study the variance in different circuits, where the term $\Vms$ is not the only significant contribution, the need to account for these dependencies would become even more evident. This necessity reaches its maximum importance when considering circuits in which modulus switching is not applied. Since the use of modulus switching is mandatory for the practical implementation of BGV, in this work we omitted the results obtained from the analysis of such circuits. 
    However, whether modulus switching is not used, or it is applied with primes that do not satisfy \Cref{gaussian.condition}, it remains crucial to take into account all dependencies arising from both $e$ and $s$ as well as from their powers. As an example of the applicability of our approach, if error coefficient distributions that are non-Gaussian, but still suitable for use, were considered, the variance analysis we propose would remain valid and could still provide accurate estimates that never underestimate the experimental values (see \Cref{teorema1}).

\end{obs}
\section{Comparison with Previous Works}
\label{sec:comparison}

When estimating the error in schemes such as BGV, existing approaches can broadly be divided into two main categories: \textit{worst-case} and \textit{average-case}. The former includes analyses based on the Euclidean norm, as in \cite{brakerski2014leveled}, the infinity norm, as in \cite{fan2012somewhat,cryptoeprint:2021/204}, and the canonical norm, as in \cite{costache2016ring,costache2020evaluating,gentry2012homomorphic,iliashenko2019optimisations,mono2022finding}, with the latter providing the tightest bounds among the worst-case analyses. In contrast, average-case approaches model the noise coefficients as random variables and focus on their mean and variance to derive probabilistic bounds on the minimum ciphertext modulus required to guarantee correctness.

Although the latter approaches appear more promising in reducing the ciphertext modulus size, research for a general method to estimate the variance remains incomplete. Most works, such as \cite{murphycentral}, treat the error coefficients as independent, which, as pointed out in \cite{BFVvar,costache2022precision,murphycentral}, leads to \textit{imprecise bounds} and often underestimates the modulus size required to prevent decryption failures.  Consequently, the scheme may be exposed to potential vulnerabilities, including key-recovery attacks \cite{cheon2024attacks}.
Other types of methods, by contrast, avoid such underestimations or failures but are limited to specific settings, as in \cite{costache2023optimisations}, where the authors analyze the error for HElib \cite{helib}.

The objective of our work is to develop an average case approach that takes into account the intrinsic dependencies introduced by the secret and public key, allowing to obtain variance estimates very close to experimental values, without ever underestimating them.
At the same time, our aim is to propose a method independent of the specific library or circuit considered, providing the theoretical basis necessary to build \textit{ad hoc} estimates depending on the circuit and library of interest.
Although, to highlight the effectiveness of our approach, we have focused on circuits that exclusively include multiplications preceded by the modulus switch, our analyses and estimates remain valid for any type of homomorphic circuit.
Our method is in fact intrinsically modular: each homomorphic operation is analyzed independently, producing variance bounds that can be composed to precisely estimate the overall noise growth of a given circuit, assuming only that the multiplications are preceded by a modulus switch -- a completely reasonable hypothesis in the context of BGV typical usage.

In this section, we demonstrate the efficacy of our average-case approach by comparing it to the state-of-the-art works. Specifically, we illustrate how our average-case analysis provides tighter and more practical bounds than traditional worst-case methods \cite{mono2022finding}, and we further compare it with other average-case approaches \cite{murphycentral,costache2023optimisations}. We then show how our choice of ciphertext moduli achieves a substantial size reduction compared to those adopted in widely used libraries such as OpenFHE \cite{OpenFHE} and HElib \cite{helib}.

\paragraph{Canonical norm.}\label{par:can.bounds} 
For the comparison with the worst-case approach, we specifically refer to the estimates proposed in \cite{mono2022finding}.

\begin{table}[htb]
    \centering
    \begin{tabular}{ll}
    \toprule
     Homomorphic operation  & Error bounds with canonical norm \\
    \midrule
    \texttt{Enc} & $||\nuenc||^\text{can} \le  D t \sqrt{n \left(1/12 + 2n  V_e V_s + V_e\right)}$\\
    \texttt{Mod Switch}$(q')$ & $||\nu + \nu_\ms(q')||^\text{can} \underset{\,}{\le} \frac{1}{p_\ell} \ncan\nu + D t \sqrt{n(1/12 + nV_s)} $ \\
    \texttt{Mult}$(\vc,\vc')$ & $||\nu_\text{mul}||^\text{can} \le ||\nu||^\text{can}  ||\nu'||^\text{can}$\\ 
    \bottomrule \\
    \end{tabular}
    \caption{Canonical norm depending on the homomorphic operations.}
    \label{tab:bgvcan}
\end{table}



One of the key distinctions between worst-case and average-case analyses lies not only in how the noise norm is bounded, but also in how these bounds propagate through homomorphic multiplication. 
In the worst-case analysis, the canonical norm is bounded as $\|\nu\|^{\text{can}} \leq D \sqrt{nV}$, 
which holds with probability at least $1 - n e^{-D^2}$, as established in \Cref{canonicalnormbound}. Thus, after one multiplication, we have (by \Cref{tab:bgvcan}) that the bound is $||\nu_\text{mul}||^\text{can} \le ||\nu||^\text{can}  ||\nu'||^\text{can}\leq D^2n\sqrt{VV'}$. Supposing that the ciphertexts are both multiplied just after the modulus switch, the bound becomes $||\nu_\text{mul}||^\text{can} \le D^2n\Vms$.

In contrast, average-case approaches, such as the one proposed in this work, allow significantly tighter bounds by imposing the condition $\|\nu\|_\infty \leq D\sqrt{2V}$, where $V$ denotes the variance of each coefficient of $\nu$.  \ec
According to \Cref{erf}, this bound holds with probability at least $1 - n(1 - \operatorname{erf}(D))$, which, for $D = 6$, exceeds $1 - 2^{-40}$. 
However, in practical scenarios, choosing $D = 8$ is preferable, as it limits the failure probability to $2^{-77}$, when $n \leq 2^{15}$.
 Note that this bound better captures the actual distribution of the error and results in much tighter predictions after multiplications. In fact, by \Cref{ourVariance}, after a multiplication that follows modulo switch, we have $||\nu_\text{mul}||_\infty \leq 2D\sqrt{n}\Vms$.\ec



\ec
\paragraph{Current Average-Case Approaches.}

Within average-case analysis frameworks, the main novelty of our approach lies in tracking the dependencies between the error coefficients, particularly those arising from the multiplication of two ciphertexts. Specifically, given the critical quantities $\nu$ and $\nu'$ of two ciphertexts with coefficient variances $V$ and $V'$, respectively, the current average-case methodology \cite{murphycentral} yields the expression $ \V(\nu \nu'|_i) \approx nVV'$.
Therefore, assuming that two ciphertexts have just undergone modulus switching, the variance resulting from their product is, in accordance with the previously established considerations \begin{equation*}
 \V(\nu \nu'|_i) \approx n\Vms^2 \,.
\end{equation*}

In contrast, our method estimates this value according to \Cref{ourVariance}
, resulting in 
\begin{equation*}
 \V(\nu \nu'|_i) \approx 2n\Vms^2\,.
\end{equation*}
The factor 2, which arises from accounting for the dependencies among the critical quantities, allows for extremely accurate values, as shown in \Cref{variance_mp_comparison}.

From this comparison, it should be evident that accounting for the dependencies in the coefficients of the error polynomial is fundamental for obtaining accurate and correct estimates of the experimental variance.\\ In particular, we believe that this consideration is precisely what overcomes the underestimation inherent in the approach presented in \cite{murphycentral}. 

It is crucial to point out that the factor of 2 is specific to the type of circuit we constructed, in which the variance is independent of the circuit level. However, when analyzing errors in circuits where this condition is not met, it is still possible to derive general upper bounds using \Cref{teorema1}, once again obtaining estimates that never underestimate the error. Nevertheless, for the reasons outlined above, including guaranteeing the Gaussianity of the error, we recommend choosing primes that allow the variance of the error after the modulus switch to be well approximated by $\Vms$, as explained previously.

Our results appear to be very close to those reported in \cite{costache2023optimisations}. In particular, it can be observed that a factor of 2 also appears in \cite[Lemma 9]{costache2023optimisations}.  
\paragraph{Current Libraries.}
Finally we compare the ciphertext modulus $q$ as estimated by our method against those 
adopted by two of the most widely used libraries: OpenFHE \cite{OpenFHE} and HElib \cite{helib}. 
We recall that the ciphertext modulus must be selected to ensure that $\|\nu\| < \frac{q}{2}$. In our selection, we additionally require \eqref{gaussian.condition} to guarantee the Gaussian behavior of the error.

\paragraph{OpenFHE.}

In OpenFHE \cite{OpenFHE}, the generation of the moduli depends on the multiplicative depth selected as input, in a manner similar to the choice of moduli described in this paper.
Moreover, the way multiplications are executed in a fixed-depth circuit of the library is consistent with the circuits analyzed in our work, where each multiplication between two ciphertexts is preceded by a modulus switching operation.
For a circuit of multiplicative depth $M$,  $M + 2 $ moduli are generated, to allow for an additional modulus switching after the final multiplication.
 Therefore, we generate the moduli in a way consistent with the library’s design, to allow a fair comparison between the modulus sizes used in OpenFHE and those proposed in this paper.

In \Cref{tab.q.35}, we compare the ciphertext modulus sizes $q$ for circuits with multiplicative depths 3 (\Cref{tab.q.3}) and 6 (\Cref{tab.q.5}). The values are obtained using our method (denoted as \textit{our}) and those generated by OpenFHE (denoted as \textit{OpenFHE}). Note that, for the generation of our parameters, we fixed $D = 8$.
\vspace{-0.5cm}
\begin{figure}[H]
    \centering
    \begin{subfigure}{.48\textwidth}
        \centering
        \resizebox{\textwidth}{!}{%
        \begin{tabular}{ccccc}
            $n$ & $2^{12}$ & $2^{13}$ & $2^{14}$ & $2^{15}$ \\
            \midrule
            OpenFHE & 147.3 & 151.8 & 156.3 & 161.6 \\
            our &  121.0  & 124.5 & 128.0 & 131.5 
        \end{tabular}}
        \caption{Circuit of depth $3$.}
        \label{tab.q.3}
    \end{subfigure}
    \hspace{0.1cm}
    \begin{subfigure}{.48\textwidth}
        \centering
        \resizebox{\textwidth}{!}{%
        \begin{tabular}{ccccc}\\
            $n$ & $2^{12}$ & $2^{13}$ & $2^{14}$ & $2^{15}$ \\
            \midrule
            OpenFHE & 249.3 & 256.8 & 264.3 & 272.6 \\
            our & 210.3 & 216.8 & 223.3 & 229.8 
        \end{tabular}}
        \caption{Circuit of depth $6$.}
        \label{tab.q.5}
    \end{subfigure}
    \caption{Comparison of $\log_2(q)$ in the circuit shown in \Cref{fig:circuit_mod} (setting  $D = 8$).}
    \label{tab.q.35}
    \medskip
\end{figure}

\vspace{-0.5cm}

From this comparison, it should be evident that our approach can lead to an effective reduction in the modulus size, resulting in corresponding improvements in the overall efficiency of the scheme.

\paragraph{HElib Comparison.}
To provide a more comprehensive overview of the effectiveness of our work and to highlight its applicability across various contexts and libraries, we also present a comparison between our choice of moduli and those used in HElib \cite{halevi2020design}, currently one of the most competitive FHE libraries, together with OpenFHE. We first point out that a strategy for optimizing HElib specific parameters has been analyzed in \cite{costache2023optimisations}. However, a direct comparison with that approach lies beyond the scope of this work. Our focus is instead on emphasizing that the strength of our method lies in its library-independent nature, offering a general framework for accurately estimating noise in different types of circuits and determining the conditions that ciphertext moduli must meet to ensure both efficiency and security (as well as the Gaussianity of the noise distributions, see \Cref{sec:BGV_Gaussian}).
For this reason, we intend to provide in this section just an idea of how our method can potentially improve current HElib parameters as well.

Our method proposes to select parameters based on the multiplicative depth of the circuit, as in OpenFHE and in previous versions of HElib. However, in the new version of the library, the selection of parameters is based on the bit length of the largest modulus in the chain. Therefore, to compare the two approaches, we focus on comparing the ratio between the moduli of successive levels, as also done in \cite{costache2023optimisations}.

Due to the way ciphertext moduli are constructed in HElib, the ratio between two moduli of adjacent levels, i.e. a prime $p_\ell$, is necessarily larger than 36 bits. However, as observed in \cite{costache2023optimisations}, it is typically much larger, often reaching sizes of around 54 bits, for a ring dimension $n \in \{2^{12},2^{13}, 2^{14}, 2^{15}\}$. Instead, with our method, we simply require that \Cref{gaussian.condition} is satisfied, i.e., for $0<l <L-1$ 
\begin{align*}
    p_\ell &> \sqrt{\frac{(2+\epsilon)n\Vms}{\alpha}} \,. 
\end{align*}
\Cref{tab.ratioHElib} reports the bit-size of the ratio between adjacent moduli resulting from our approach, for ring dimensions in $\{ 2^{12},2^{13},2^{14},2^{15} \}$, assuming a secret key with coefficients sampled from a ternary uniform distribution, with hamming weight $h = \sfrac{n}{2}$, plaintext modulus $t= 65537$ and $\alpha =  \sfrac{1}{100}$. As shown in the table, our approach achieves a reduction of up to 6 bits in the ratio between consecutive ciphertext moduli.
\begin{table}[h!]
\centering
\begin{tabular}{ccccc}
\toprule
 $n$ & $2^{12}$ & $2^{13}$ & $ 2^{14}$ & $2^{15}$ \\
\midrule
$\log_2(p_\ell)$ & 29.55 & 30.55 & 31.55  & 32.55 \\
\bottomrule \\
\end{tabular}
\caption{\footnotesize Ratio between adjacent ciphertext moduli for different ring dimensions $n$, according to our approach, represented using their bit size.}
\label{tab.ratioHElib}
\end{table}

\vspace{-1cm}
\section{Conclusions}
\label{sec:BGVconclusions}

In this work, we introduce a new approach to average cases that can accurately estimate the error arising from homomorphic operations, especially multiplication. The main novelty of our method lies in the definition of error bounds that account for the dependencies between the critical quantities of multiplied ciphertexts, arising from both the public and private keys. We believe that this is precisely what enables our approach to overcome the typical underestimations observed in current average-case analyses, thereby providing accurate bounds on the error variance without ever underestimating it.

Furthermore, this work aims to provide general guidelines for studying noise growth in circuits independently of the specific implementation of the homomorphic encryption library employed. 
Based on these estimates, we present a method to select ciphertext moduli appropriately in a generic circuit, demonstrating that the new solution facilitates a significant reduction in their size. This leads to improved efficiency compared to current techniques and compared to the moduli employed in major homomorphic encryption libraries, with no loss of security and proper scheme functioning.

Finally, we provide theoretical and empirical evidence supporting the validity of average-case approaches in the study of homomorphic schemes such as BGV. Precisely, we show that a proper choice of moduli together with application of modulus switching on a systematic basis in BGV allow error distributions to be well approximated by Gaussian distributions. This confirms the validity of such approaches and highlights their practical relevance in improving efficiency, thereby contributing to the potential widespread adoption of these schemes in real-world applications.


\section*{Acknowledgment}
The third and fourth authors were partially supported by the Italian Ministry of
University and Research in the framework of the Call for Proposals for
scrolling of final rankings of the PRIN 2022 call - Protocol no.
2022RFAZCJ. The third author acknowledges project SERICS (PE00000014) under the MUR National Recovery and Resilience Plan funded by the European Union - NextGenerationEU.

\bibliographystyle{splncs04}
\bibliography{citations}

\appendix
\appendix
\label{appendix}

\section{Proof of \Cref{lemma0} }
\label{proof.lemma0}

\begin{description}
\item In order to prove this lemma, we will first demonstrate that these properties hold for the critical quantity of a fresh ciphertext $\nuenc$.\\ Then, we will show that any operation involved in the BGV circuit does not affect these properties.\\
\item[Fresh ciphertexts]
For $\nuenc$, the coefficients $b_{\iota}(\mu)$ are defined as
\begin{equation*}
\begin{cases}
    b_0(0) = m + te_0  \\
    b_1(0) = tu 
\end{cases}
\begin{cases}
    b_0(1) = te_1  
\end{cases}
\end{equation*}  
Therefore, the first property follows immediately from the independence of $b_{\mu_1}(\iota_1)|_{j_1},b_{\mu_2}(\iota_2)|_{j_2}$ when $\mu_1 \neq \mu_2 \text{ or } j_1 \neq j_2$.\\\\
As for the second property, it holds since
 \begin{itemize}
    \item $\E[b_0(0)|_i] = \E[m|_i] + t \E[e_0|_i] = 0$, due to the linearity of the expected value and the distributions considered, i.e. $m \leftarrow \U_t, e_0 \leftarrow \mathcal{DG}_{q}(\sigma^2)$;
     \item $\E[b_1(0)|_i] = t \E[u|_i] = 0$, as $u \leftarrow \chi_s$;
     \item $\E[b_0(1)|_i]= t\E[e_1|_i] = 0 $, as $e_1 \leftarrow  \mathcal{DG}_{q}(\sigma^2)$; \\
\end{itemize}
 We will therefore show that the remaining homomorphic operations do not alter these properties.\\\\ Let $\nu = \sum_{\iota_1} \sum_{\mu_1} b_{\mu_1}(\iota_1) e^{\mu_1} s^{\iota_1}$ , $\nu' = \sum_{\iota_2} \sum_{\mu_2} b'_{\mu_2}(\iota_2) e^{\mu_2} s^{\iota_2} $ be the respective critical quantities of two generic ciphertexts, for which the properties stated above are assumed to hold.\\
\item[Addition of two ciphertexts]  The critical quantity after the addition of the two BGV ciphertexts is given by \begin{equation*}
     \nu_{\text{add}} = \nu + \nu' = \sum_{\iota} \sum_{\mu} b^{\text{add}}_{\mu}(\iota) e^{\mu} s^{\iota} ,
\end{equation*}
where $b^{\text{add}}_{\mu}(\iota) = b_{\mu}(\iota) + b'_{\mu}(\iota)$. \\\\ Therefore, if $\E[ b_{\mu}(\iota)|_i] = \E[ b'_{\mu}(\iota)|_i] = 0$ then \begin{equation*}
    \E[ b^{\text{add}}_{\mu}(\iota)|_i] = 0,
\end{equation*}
according to the linearity of the expected value.\\ \\
Moreover, using the bilinearity of the covariance, we have that, for $\mu_1 \neq \mu_2 \text{ or } j_1 \neq j_2$ \begin{equation*}
    \Cov(b_{\mu_1}^{\text{add}}(\iota_1)|_{j_1},b_{\mu_2}^{\text{add}}(\iota_2)|_{j_2})= 0 ,
\end{equation*}
since \begin{align*}
    \Cov( b_{\mu_1}(\iota_1) + b'_{\mu_1}(\iota_1)|_{j_1},b_{\mu_2}(\iota_2) + b'_{\mu_2}(\iota_2)|_{j_2}) &= \Cov(b_{\mu_1}(\iota_1)|_{j_1},b_{\mu_2}(\iota_2)|_{j_2})  \\  &\quad+  \Cov(b_{\mu_1}(\iota_1)|_{j_1},b'_{\mu_2}(\iota_2)|_{j_2})  \\  &\quad+ \Cov(b'_{\mu_1}(\iota_1)|_{j_1},b_{\mu_2}(\iota_2)|_{j_2})  \\  &\quad+ \Cov(b'_{\mu_1}(\iota_1)|_{j_1},b'_{\mu_2}(\iota_2)|_{j_2}),
\end{align*}
where all the summands vanish because: \begin{itemize}
    \item $\Cov(b_{\mu_1}(\iota_1)|_{j_1},b_{\mu_2}(\iota_2)|_{j_2})=\Cov(b'_{\mu_1}(\iota_1)|_{j_1},b'_{\mu_2}(\iota_2)|_{j_2})=0$ holds by assumption for $\mu_1 \neq \mu_2 $ or $ j_1 \neq j_2$;
    \item $\Cov(b_{\mu_1}(\iota_1)|_{j_1},b'_{\mu_2}(\iota_2)|_{j_2}) =  \Cov(b'_{\mu_1}(\iota_1)|_{j_1},b_{\mu_2}(\iota_2)|_{j_2})=0$ since $b_{\mu_1}(\iota_1)$ and $b'_{\mu_2}(\iota_2)$ are independent $\forall \mu_1, \mu_2$;\\
\end{itemize} 
\item[Multiplication by a constant] Given a ciphertext with $\nu = \sum_{\iota} \sum_{\mu} b_{\mu}(\iota) e^{\mu} s^{\iota}  $ and a constant $\alpha$, the critical quantity obtained after their homomorphic multiplication, according to the BGV scheme, is as follows  \begin{equation*}
     \nu_{\text{const}} = \alpha\nu =\alpha \sum_{\iota} \sum_{\mu} b_{\mu}(\iota) e^{\mu} s^{\iota} = \sum_{\iota} \sum_{\mu} \alpha b_{\mu}(\iota) e^{\mu} s^{\iota}.
\end{equation*}
Therefore, we can define $b^{\text{const}}_{\mu}(\iota) = \alpha b_{\mu}(\iota)$ from which, according to the linearity of the expected value \begin{equation*}
     \E[ b^{\text{const}}_{\mu}(\iota)|_i] =  \E[ \alpha b_{\mu}(\iota)|_i]= \alpha \E[ b_{\mu}(\iota)|_i] = 0,
\end{equation*}
which proves property $b)$.\\ \\
Property $a)$ follows directly by assumption from the bilinearity of the covariance \begin{align*}
     \Cov(b_{\mu_1}^{\text{const}}(\iota_1)|_{j_1},b_{\mu_2}^{\text{const}}(\iota_2)|_{j_2}) &= \Cov(\alpha b_{\mu_1}(\iota_1)|_{j_1},\alpha b_{\mu_2}(\iota_2)|_{j_2}) = \\ &= \alpha^2 \Cov( b_{\mu_1}(\iota_1)|_{j_1}, b_{\mu_2}(\iota_2)|_{j_2}) = 0.
\end{align*}
\item[Multiplication of two ciphertexts]
The critical quantity arising from the multiplication of two ciphertexts, whose associated noise is defined as above, can be expressed as
\begin{equation*}
    \numul = \nu \cdot \nu' =  \sum_{\iota} a^{\mathsf{mul}}_{\iota}s^{\iota} =  \sum_{\iota} \sum_{\mu } b^{\mathsf{mul}}_{\mu}({\iota}) e^{\mu} s^{\iota},
\end{equation*}
where $a^{\mathsf{mul}}_{\iota} = \sum_{\iota_1+\iota_2=\iota}a_{\iota_1} a'_{\iota_2}$. \\\\ Moreover \[
\begin{cases}
    a_{\iota_1} = \sum_{\mu_1} b_{\mu_1}({\iota_1}) e^{\mu_1} \\
    a'_{\iota_2} = \sum_{\mu_2} b'_{\mu_2}({\iota_2}) e^{\mu_2}
\end{cases}
\]
From which it follows that \begin{equation*}
    b^{\mathsf{mul}}_{\mu}({\iota})|_i =  \sum_{{\iota_1}+{\iota_2}=\iota} \sum_{\mu_1+\mu_2=\mu} \sum_j b_{\mu_1}({\iota_1})|_j \, b'_{\mu_2}({\iota_2})|_{i-j}.
\end{equation*} 

From the independence of $b_{\mu_1}(\iota_1)$ and  $b'_{\mu_2}(\iota_2)$ $\forall \mu_1, \mu_2,\iota_1, \iota_2$, and for the linearity of the expected value, one can deduce that \begin{align*}
    \E[b^{\mathsf{mul}}_{\mu}({\iota})|_i] &= \sum_{{\iota_1}+{\iota_2}=\iota} \sum_{\mu_1+\mu_2=\mu} \sum_j\E[b_{\mu_1}({\iota_1})|_j \, b'_{\mu_2}({\iota_2})|_{i-j}] \\ &=\sum_{{\iota_1}+{\iota_2}=\iota} \sum_{\mu_1+\mu_2=\mu} \sum_j\E[b_{\mu_1}({\iota_1})|_j ] \E[b'_{\mu_2}({\iota_2})|_{i-j}] = 0,
\end{align*}
which easily proves property $b)$. \\\\ The expression of the covariance $\Cov(b^{\mathsf{mul}}_{\mu_1}({\iota_1})|_{i_1}, b^{\mathsf{mul}}_{\mu_2}({\iota_2})|_{i_2})$ can be reduced, using its bilinearity, to a sum of terms of the form 
\begin{align*}
      \Cov(b_{\mu_1}({\iota_1})|_{l_1}b'_{\mu_2}({\iota_2})|_{i_1-l_1},b_{\mu_3}({\iota_3})|_{l_2}b'_{\mu_4}({\iota_4})|_{i_2-l_2} ) ,
\end{align*}
which are all zero, using the property of the covariance stated below. \begin{property}\label{cov}
    Let $X_1,X_2,X_3,X_4$ be some fixed random variables.\\ If $X_2,X_4$ are independent with respect to $X_1,X_3$, $\Cov(X_2,X_4) = 0$ and   $\E[X_2] = 0$ then \begin{equation*}
        \Cov( X_1 \cdot X_2, X_3 \cdot X_4 ) = 0.
    \end{equation*} 
\end{property}
\noindent Thus, property $a)$ follows by observing that, for $\mu_2 \neq \mu_4 \text{ or } i_2 \neq i_4$: \begin{itemize}
    \item $\Cov(b'_{\mu_2}({\iota_2})|_{i_1-l_1}, b'_{\mu_4}({\iota_4})|_{i_2-l_2}) = 0$ e $\E[b'_{\mu_2}({\iota_2})|_{i_1-l_1}]=0$ based on the hypotheses made;
    \item $b'_{\mu_2}({\iota_2})|_{i_1-l_1}, b'_{\mu_4}({\iota_4})|_{i_2-l_2}$ are independent with respect to $b_{\mu_1}({\iota_1})|_{l_1}, b_{\mu_3}({\iota_3})|_{l_2}$;
\end{itemize}

Therefore, thanks to the property \ref{cov}, this implies that  \begin{equation*}
    \Cov(b_{\mu_1}({\iota_1})|_{l_1}b'_{\mu_2}({\iota_2})|_{i_1-l_1},b_{\mu_3}({\iota_3})|_{l_2}b'_{\mu_4}({\iota_4})|_{i_2-l_2}) = 0.
\end{equation*}
\item[Modulus and Key Switching]
Let $\bm{c}= (c_0,c_1)$ be a ciphertext in $R_{q_l}\times R_{q_l}$ and suppose that the modulus switch to $q_{l'}$ is applied, in order to reduce the error. 

The resulting ciphertext is defined as \begin{equation*}
    \bm{c'}= \frac{q_{l'}}{q_l}(\bm{c} + \bm{\delta}) \mod q_{l'},
\end{equation*}
where $\bm{\delta} = t[-\bm{c}t^{-1}]_{\frac{q_l}{q_{l'}}}$.

The critical quantity associated to $\bm{c'}$ can be expressed as \begin{equation*}
    \nu' = \frac{q_{l'}}{q_l}(\nu + \nu_{\mathsf{ms}}) \quad \text{ where} \quad \nu_{\mathsf{ms}} = \delta_0 + \delta_1s \,.
\end{equation*}
It is possible to observe that the ciphertext components ${c}_0,c_1 $ can be thought as randomly distributed over $R_{q_l}$, ${c_0},c_1 \leftarrow \mathcal{U}_{q_l}$, and therefore the $\delta_i$ can be treated as independent polynomials with coefficients chosen randomly over $I= {(-\frac{tq_l}{2q_{l'}},\frac{tq_l}{2q_{l'}})} $, i.e. $\delta_0,\delta_1 \leftarrow \mathcal{U}_I$.\\ Moreover, it should be noted that the values $\delta_i$ exclusively influence $b_0(0), b_0(1)$, and that they have an expected value equal to zero, because of their distributions. 

Therefore, referring back to the case of the homomorphic sum, we can deduce that the expected value of \( b'_{\mu}(\iota) \) for the new ciphertext $\bm{c'} $ remains zero, as do the covariances. 

In the same way, by reducing the problem to the case of homomorphic addition, it is possible to show that these properties remain valid also after the relinearization process.

We decided not to report all the technical details but to provide only the key underlying idea, as there are multiple relinearization variants and including them would have required too much space. However, all the calculations can be derived in a very straightforward manner by simply adapting the approach in \cite{BFVvar}.
\end{description}

\section{Proof of \Cref{lemma1} }
\label{proof.lemma1}

In order to prove the statement, we start by writing the term $a_{\iota}s^{\iota}|_i$, according to \ref{productpoly},  as \begin{equation*}
    a_{\iota}s^{\iota}|_i = \sum_{\mu} (b_{\mu}(\iota)e^{\mu}s^{\iota})|_i = \sum_{\mu} \sum_{j=0}^{n-1} \xi(i,j) b_{\mu}(\iota)e^{\mu}|_j s^{\iota}|_{i-j}.
\end{equation*} 
Thus, given two random variables $X$ and $Y$, the following properties hold:
\begin{align*}
    &\text{a. \hspace{0.2cm}} \V(XY) = (\V(X) + \E[X]^2)(\V(Y) + \E[Y]^2) + \Cov(X^2,Y^2) \\
               &\hspace{2.4cm} - (\Cov(X,Y) + \E[X] \E[Y])^2 \\
               &\text{b. \hspace{0.2cm}}  \V(X+Y) = \V(X) + \V(Y) +2\Cov(X,Y),
\end{align*}
where the second property can be generalized for $k$ random variables $\{X_i\}_{i=0}^k$~as 
\begin{equation*}
    \V\left(\sum_{i=0}^k X_i\right) = \sum_{i=0}^k \V(X_i) + \sum_{i_1 \ne i_2} \Cov(X_{i_1},X_{i_2}).
\end{equation*}
Then, it is possible to compute the variance of $a_{\iota}s^{\iota}|_i$ as \begin{align*}
    \V(a_{\iota}s^{\iota}|_i) &= \sum_{\mu} \sum_{j=0}^{n-1} \V(b_{\mu}(\iota)e^{\mu}|_j s^{\iota}|_{i-j})  \\ &+ \sum_{\mu_1 \neq \mu_2 \text{ or } j_1 \neq j_2}  \xi(i,j_1) \xi(i,j_2) \Cov(b_{\mu_1}(\iota)e^{\mu_1}|_{j_1} s^{\iota}|_{i-j_1}, b_{\mu_2}(\iota)e^{\mu_2}|_{j_2} s^{\iota}|_{i-j_2} ),
\end{align*}
where the covariances vanishes based on property \ref{cov}. Thus, the following equality holds
\begin{equation}\label{aiota}
    \V(a_{\iota}s^{\iota}|_i) = \sum_{\mu} \sum_{j=0}^{n-1} \V(b_{\mu}(\iota)e^{\mu}|_j s^{\iota}|_{i-j}).
\end{equation} 
Moreover, according to property (a.), it follows that 
    \begin{align*}
     \V(b_{\mu}(\iota)e^{\mu}|_j s^{\iota}|_{i-j}) &= (\V(b_{\mu}(\iota)e^{\mu}|_j )+ \E[b_{\mu}(\iota)e^{\mu}|_j]^2)(\V(s^{\iota}|_{i-j}) + \E[s^{\iota}|_{i-j}]^2)  \\ &\quad+ \Cov(b_{\mu}(\iota)e^{\mu}|_j^2,s^{\iota}|_{i-j}^2) \\ &\quad- \left(\Cov(b_{\mu}(\iota)e^{\mu}|_j,s^{\iota}|_{i-j}) +\E[b_{\mu}(\iota)e^{\mu}|_j]\E[s^{\iota}|_{i-j}] \right)^2.
\end{align*} 
At this point, it should be noted that \begin{itemize}
    \item $\E[b_{\mu}(\iota)e^{\mu}|_j] = 0$ according to \Cref{lemma0};
    \item $\Cov(b_{\mu}(\iota)e^{\mu}|_j,s^{\iota}|_{i-j}) = \Cov(b_{\mu}(\iota)e^{\mu}|_j^2,s^{\iota}|_{i-j}^2) = 0 $ as $b_{\mu}(\iota)e^{\mu}|_j,s^{\iota}|_{i-j}$ are independent;
\end{itemize}
This results in
$\V(b_{\mu}(\iota)e^{\mu}|_j s^{\iota}|_{i-j}) = \V(b_{\mu}(\iota)e^{\mu}|_j )(\V(s^{\iota}|_{i-j}) + \E[s^{\iota}|_{i-j}]^2)$. 
In addition, for a random variable $X$, it holds that 
$\V(X)= \E[X^2]-\E[X]^2$.
Therefore, 
\begin{equation}\label{bmu}
     \V(b_{\mu}(\iota)e^{\mu}|_j s^{\iota}|_{i-j}) = \V(b_{\mu}(\iota)e^{\mu}|_j )\E[s^{\iota}|_{i-j}^2].
\end{equation}
Finally, the same reasoning can be applied in order to derive \begin{equation}\label{bmusum}
    \V(b_{\mu}(\iota)e^{\mu}|_j ) = \sum_{k=0}^{n-1} \V(b_{\mu}(\iota)|_k)  \E[e^{\mu}|_{j-k}^2].
\end{equation}
In fact, using \eqref{productpoly} and property (b.), it follows that 
\begin{align*}
     \V(b_{\mu}(\iota)e^{\mu}|_j) &= \V\left(\sum_{k=0}^{n-1}b_{\mu}(\iota)|_ke^{\mu}|_{j-k}\right)\\
     &= \sum_{k=0}^{n-1}\V(b_{\mu}(\iota)|_ke^{\mu}|_{j-k}) \\&\hspace{0.4cm}+ \sum_{ k_1 \neq k_2}  \xi(j,k_1) \xi(j,k_2) \Cov(b_{\mu}(\iota)|_{k_1} e^{\mu}|_{j-k_1}, b_{\mu}(\iota)|_{k_2} e^{\mu}|_{j-k_2} ),
\end{align*} 
where the covariances are null thanks to property \ref{cov}.\\ Moreover, according to property (a.), it follows that \begin{align*}
    \V(b_{\mu}(\iota)|_ke^{\mu}|_{j-k}) &= (\V(b_{\mu}(\iota)|_k)+ 
     \E[b_{\mu}(\iota)|_k]^2)(\V(e^{\mu}|_{j-k}) + \E[e^{\mu}|_{j-k}]^2)  \\ &\quad+ \Cov(b_{\mu}(\iota)|_k^2,e^{\mu}|_{j-k}^2) \\ &\quad- (\Cov(b_{\mu}(\iota)|_k,e^{\mu}|_{j-k}) +\E[b_{\mu}(\iota)|_k]\E[e^{\mu}|_{j-k}] )^2.
\end{align*}
Thus, \eqref{bmu} is proven observing that $\E[b_{\mu}(\iota)|_k] = 0$, according to lemma \ref{lemma0}, and that $b_{\mu}(\iota)|_k$ and $e^{\mu}|_{j-k}$ are independent.\\
By substituting \eqref{bmu} and \eqref{bmusum} in \eqref{aiota}, it follows that \begin{align*}
    \V(a_{\iota}s^{\iota}|_i) &= \sum_{\mu} \sum_{j=0}^{n-1} \V(b_{\mu}(\iota)e^{\mu}|_j s^{\iota}|_{i-j}) = \sum_{\mu} \sum_{j=0}^{n-1}  \V(b_{\mu}(\iota)e^{\mu}|_j )\E[s^{\iota}|_{i-j}^2] \\ &= \sum_{\mu} \sum_{j=0}^{n-1}  \sum_{k=0}^{n-1} \V(b_{\mu}(\iota)|_k)  \E[e^{\mu}|_{j-k}^2]\E[s^{\iota}|_{i-j}^2]. \\ 
\end{align*}
Finally, observing that $\V(b_{\mu}(\iota)|_i), \E[e^{\mu}|_{i}^2]$ and $\E[s^{\iota}|_{i}^2]$ do not depend on $i$, the thesis is demonstrated, i.e., \begin{equation*}
    \V(a_{\iota}s^{\iota}|_i) = \sum_{\mu} \V(b_{\mu}(\iota)|_i)  \sum_{k=0}^{n-1}\E[e^{\mu}|_{k}^2]\sum_{j=0}^{n-1}\E[s^{\iota}|_{j}^2].
\end{equation*}

The expression for the variance of the coefficients of the critical quantity simply follows from the observation that the covariance between $a_{\iota_1}s^{\iota_1}|_i$ and $a_{\iota_2}s^{\iota_2}|_i$ vanishes whenever $\iota_1 \ne \iota_2$.  Indeed, by bilinearity of the covariance operator, $\Cov(a_{\iota_1}s^{\iota_1}|_i,a_{\iota_2}s^{\iota_2}|_i)$ can be expressed as a sum of terms of the form \begin{equation*}
    \Cov(b_{\mu_1}({\iota_1})|_j(e^{\mu_1}s^{\iota_1})|_{i-j},b_{\mu_2}({\iota_2})|_j(e^{\mu_2}s^{\iota_2})|_{i-j}) \,,
   \end{equation*}
   each of which vanishes by Property~\ref{cov}. Therefore, we obtain
   \begin{equation*}
       \V(\nu|_i) = \V(\sum_{\iota\ge 0} a_{\iota}s^{\iota}|_i) = \sum_{\iota\ge 0}\V(a_{\iota}s^{\iota}|_i) + \sum_{\iota_1 \ne \iota_2}\Cov(a_{\iota_1}s^{\iota_1}|_i,a_{\iota_2}s^{\iota_2}|_i) \,,
   \end{equation*}
   thus concluding the proof.

\section{On the Negligibility of Key Switching}\label{KS.negligibility}
In this work, we chose to omit and not explicitly address the contribution of key switching in our analysis of noise variance, despite its necessity for ensuring the practicality of evaluated circuits. This decision is justified by the fact that, in general, the parameters associated with key switching can be selected so that its impact on the overall noise remains negligible in comparison to that of other operations — especially the multiplication operation, which immediately precedes the key switching step.

While our approach can easily accommodate this contribution, we opted to exclude it from our presentation for the sake of clarity and due to the aforementioned reason. Various key switching methods exist. To support the reasonableness of our choice, we provide a justification of its negligible impact for one of the main variants: GHS (Gentry Halevi Smart).

Let $q$ be the modulus of the ciphertext to be relinearized, and let $Q$ be a modulus specifically chosen for this purpose, such that  $q < Q$ and $q \mid Q$. \\The core idea of the GHS variant is to perform relinearization in the larger ring $\R_Q$, and then apply modulus switching to return to the smaller ring $\R_q$, reducing the error.\\

Let $c=(c_0,c_1,c_2)$ be the result of the multiplication of two ciphertexts, with $c_0,c_1,c_2 \in \R_q$. \\The key-switching key for the GHS variant is defined as 
\begin{equation*}
    ({ek_0}, {ek_1}) = \left( \left( -a \cdot s + t e + \frac{Q}{q} s^2 \right), a \right) \mod Q \,.
\end{equation*}
The ciphertext resulting from relinearization in $R_Q$ is given by \begin{equation*}
\begin{split}
     {c'_0} = \left[ \frac{Q}{q}c_0+c_2 \cdot ek_0 \right]_Q \,, \quad   {c'_1} = \left[\frac{Q}{q}c_1+c_2 \cdot ek_1 \right]_Q \,,
\end{split}
\end{equation*}
Then, this new ciphertext is scaled back modulo $q$ using the modulus switching technique previously introduced, thereby obtaining \begin{equation*}
\begin{split}
     \hat{c_0} =  \left[\frac{q}{Q}({c'_0} + \delta_0)\right]_q \hspace{1cm} \text{where } \delta_0 = t[-c'_0t^{-1}]_{\frac{Q}{q}}\,, \\
    \hat{c_1} =  \left[\frac{q}{Q}({c'_1} + \delta_1)\right]_q \hspace{1cm} \text{where } \delta_1 = t[-c'_1t^{-1}]_{\frac{Q}{q}} \,,
\end{split}
\end{equation*} 
It is easy to verify that the critical quantity of the new relinearized ciphertext is given by 
\begin{equation*}
\begin{aligned}
   \nu_{\mathsf{ks}} = [\hat{c_0} + \hat{c_1} \cdot s]_q 
    = \numul + \left[ \frac{q}{Q} \left( t c_2 \cdot e + \delta_0 + \delta_1 \cdot s \right) \right]_q \,.
\end{aligned}
\end{equation*}
A typical choice for the size of $Q$ is  $Q \approx q^2$ \cite{cryptoeprint:2021/204}.
In this way, the variance of the second component will be 
\begin{align*}
    V_{\frac{q}{Q}(tc_2\cdot e + \delta_0 + \delta_1 \cdot s)} = \frac{1}{q^2} \left({nt^2V_eV_{c_2}} + { V_{\delta_0} + nV_sV_{\delta_1}} \right)= \frac{t^2}{12} \left({nV_e} + { 1 + nV_s} \right) \,.
\end{align*}
On the other hand, in accordance with the circuits proposed in the paper, we may assume that the ciphertext to be relinearized results from the product of two terms that have already undergone modulus switching.\\ Consequently, the variance of the first component, namely $\numul$, is at least $n\Vms^2 F_s(1,1)$ where $\Vms=\frac{t^2nV_s}{12}$, thereby making evident the negligible contribution of the key switching step.

\end{document}